\documentclass{siamltex}

\usepackage{algorithm,algorithmic}
\usepackage{amsmath,amssymb}
\usepackage{cite}
\usepackage{color}
\usepackage{graphicx}
\usepackage{rotating}
\usepackage{url}
\usepackage{verbatim}
\usepackage{xspace}

\title{Parallel tree algorithms for AMR
       \\
       and non-standard data access%
       \thanks{This work was supported by
              the Hausdorff Center for Mathematics,
              Universit\"at Bonn, Germany}}%

\author{Carsten Burstedde\thanks{Institut f\"ur Numerische Simulation,
        Rheinische Friedrich-Wilhelms-Universit\"at Bonn, Germany}}%

\newtheorem{convention}[theorem]{Convention}
\newtheorem{property}[theorem]{Property}
\newtheorem{principle}[theorem]{Principle}

\newcommand{\dendro}{\texttt{Dendro}\xspace}
\newcommand{\pforest}{\texttt{p4est}\xspace}
\newcommand{\pforestfun}[1]{\texttt{p4est\_\-#1}\xspace}
\newcommand{\pfbuild}{\pforestfun{build}}
\newcommand{\pfcompleteregion}{\pforestfun{complete\_\-region}}
\newcommand{\pfcompletesubtree}{\pforestfun{complete\_\-subtree}}

\newcommand{\pfghostexpand}{\pforestfun{ghost\_\-expand}}
\newcommand{\pfnearestcommon}{\pforestfun{nearest\_\-common\_\-ancestor}}
\newcommand{\pfenlargefirst}{\pforestfun{enlarge\_\-first}}
\newcommand{\pfenlargelast}{\pforestfun{enlarge\_\-last}}
\newcommand{\pfpartition}{\pforestfun{partition}}
\newcommand{\pfpertree}{\pforestfun{count\_\-pertree}}
\newcommand{\pfsave}{\pforestfun{save}}
\newcommand{\pfsearch}{\pforestfun{search}}

\newcommand{\pfsearchpartition}{\pforestfun{search\_\-partition}}

\newcommand{\pftransferfixed}{\pforestfun{transfer\_\-fixed}}
\newcommand{\pftransfervariable}{\pforestfun{transfer\_\-variable}}

\newcommand{\narynotify}{\texttt{sc\_nary\_\-notify}\xspace}
\newcommand{\tetcode}{\texttt{t8code}\xspace}

\newcommand{\scfun}[1]{\texttt{sc\_\-#1}\xspace}
\newcommand{\scsplit}{\scfun{array\_\-split}}

\newcommand{\fun}[1]{{\normalfont \texttt{#1}}\xspace}
\newcommand{\mpifun}[1]{\texttt{MPI\_\-#1}\xspace}
\newcommand{\mpiallgather}{\mpifun{Allgather}}
\newcommand{\mpiallgatherv}{\mpifun{All\-gatherv}}
\newcommand{\mpialltoall}{\mpifun{Alltoall}}

\newcommand{\mpiirecv}{\mpifun{Irecv}}

\newcommand{\Add}{\fun{Add}}
\newcommand{\Match}{\fun{Match}}

\newcommand{\beginsquadrant}{\fun{begins\_\-with}}
\newcommand{\begintree}{\fun{begin\_tree}}
\newcommand{\ndtree}{\fun{end\_tree}}
\newcommand{\buildadd}{\pforestfun{build\_add}}
\newcommand{\buildnew}{\pforestfun{build\_begin}}
\newcommand{\buildend}{\pforestfun{build\_end}}

\newcommand{\rcp}{RC+P\xspace}
\newcommand{\rcbs}{RC+B\xspace}
\newcommand{\rcbsp}{RC+B+P\xspace}
\newcommand{\ie}{i.e.\xspace}

\newcommand{\alglab}[1]{\label{alg:#1}}
\newcommand{\algref}[1]{Algorithm~\ref{alg:#1}\xspace}

\newcommand{\eqnlab}[1]{\label{eqn:#1}}
\newcommand{\eqnref}[1]{\eqref{eqn:#1}\xspace}
\newcommand{\figlab}[1]{\label{fig:#1}}
\newcommand{\figref}[1]{Figure~\ref{fig:#1}\xspace}
\newcommand{\tablab}[1]{\label{tab:#1}}
\newcommand{\tabref}[1]{Table~\ref{tab:#1}\xspace}
\newcommand{\linelab}[1]{\label{line:#1}}
\newcommand{\lineref}[1]{Line~\ref{line:#1}\xspace}
\newcommand{\conventionref}[1]{Convention~\ref{convention:#1}\xspace}
\newcommand{\propertyref}[1]{Property~\ref{property:#1}\xspace}

\newcommand{\principleref}[1]{Principle~\ref{principle:#1}\xspace}
\newcommand{\seclab}[1]{\label{sec:#1}}
\newcommand{\secref}[1]{Section~\ref{sec:#1}\xspace}

\newcommand{\arrE}{\mathfrak{E}}
\newcommand{\arrK}{\mathfrak{K}}

\newcommand{\arrN}{\mathfrak{N}}
\newcommand{\arrO}{\mathfrak{O}}
\newcommand{\arra}{\mathfrak{a}}

\newcommand{\arrm}{\mathfrak{m}}
\newcommand{\arrn}{\mathfrak{n}}

\newcommand{\rmin}{\mathrm{min}}
\newcommand{\rmax}{\mathrm{max}}
\newcommand{\rgeom}{\mathrm{geom}}
\newcommand{\rarith}{\mathrm{arith}}
\newcommand{\rafter}{\mathrm{after}}
\newcommand{\rbefore}{\mathrm{before}}
\newcommand{\rdesc}{\mathrm{desc}}
\newcommand{\rfirst}{\mathrm{first}}
\newcommand{\rfirstlocaltree}{\mathrm{first\_local\_tree}}

\newcommand{\rlast}{\mathrm{last}}
\newcommand{\rlastlocaltree}{\mathrm{last\_local\_tree}}
\newcommand{\rnumbers}{\mathrm{numbers}}
\newcommand{\rroot}{\mathrm{root}}
\newcommand{\rtree}{\mathrm{tree}}
\newcommand{\rtrees}{\mathrm{trees}}
\newcommand{\relements}{\mathrm{elements}}
\newcommand{\rinval}{\mathrm{is\_invalid}}
\newcommand{\rprev}{\mathrm{most\_recently\_added}}

\newcommand{\dr}{\mathrm{d}r}

\newcommand{\Ttree}{\mathcal{K}}

\DeclareMathOperator{\bor}{\vert}
\DeclareMathOperator{\band}{\&}

\newcommand{\mminus}{\hspace{.1ex}\text{$--$}}
\newcommand{\pplus}{\hspace{.1ex}\text{$++$}}

\newcommand{\cO}{\mathcal{O}}
\newcommand{\sR}{\mathbb{R}}
\newcommand{\sZ}{\mathbb{Z}}

\newcommand{\closedin}[1]{\lbrack#1\rbrack}
\newcommand{\halfopen}[1]{\lbrack#1)}

\begin{document}

\maketitle

\begin{abstract}
  We introduce several
  % data-oriented
  % highly scalable
  parallel
  algorithms operating on a distributed forest of adaptive quadtrees/octrees.
  They are targeted at large-scale applications relying on
  % a non-standard
  data layouts that are more complex than required for standard finite elements,
  % .
  % Such applications appear in various contexts,
  % examples being
  such as
  $hp$-adaptive
  % discontinuous
  Galerkin methods,
  % computational
  % element-based
  particle tracking and semi-Lagrangian schemes,
  % of large numbers of moving points,
  and in-situ post-processing and visualization.
  Specifically, we design algorithms to derive an adapted worker forest
  based on sparse data, to identify owner processes in
  a top-down search of remote objects,
  and to allow for variable process counts and per-element data sizes in
  partitioning and parallel file I/O.
  % The latter supports optimized searches of remote objects and their associated
  % processes.
  %To accelerate applications relying on these algorithms, we develop
  %% several
  %communication schemes that minimize the number of point-to-point messages by
  %exploiting the properties of
  % the
  %a space filling curve.
  % They are useful for partition-independent loading and saving of the mesh in
  % one flat file and to transfer fixed- and variable-size element data during
  % repartitioning.
%
%   While we use the \pforest metadata definition as a reference, these
%   algorithms are applicable to any distributed forest with linear element
%   storage, including the traditional
%   % case of a
%   one-tree forest.
  We demonstrate the algorithms' usability and performance in the context of
  % We describe their use by several example applications.
  % aid in modularizing
  % an in-situ visualization
  % several computational examples.
  a particle tracking example
  that we scale to 21e9 particles and 64Ki MPI processes on the Juqueen
  supercomputer,
  and we describe the construction of a parallel assembly of variably sized
  spheres in space creating up to 768e9 elements on the Juwels supercomputer.

\end{abstract}



\pagestyle{myheadings}
\thispagestyle{plain}
\markboth{C.\ BURSTEDDE}{TREE ALGORITHMS FOR NON-STANDARD DATA}

\section{Introduction}

% This paper concerned non-standard.  Data.  To make this more clear,
% let us first summarize what we consider as standard.

% approximate, convergence.

Numerical methods to solve partial differential equations (PDEs) have become
ubiquitous in science and industry.
Many approaches subdivide the domain of the PDE into a mesh of cells that
constitute the computational elements.
The finite/spectral element/volume methods are among the most prevalent
techniques and establish mathematical links between nearest-neighbor elements;
% , and sometimes neighbor elements of a larger, fixed distance.
% This holds for both uniform and adaptive meshes
see e.g.\ \cite{StrangFix88, BabuskaGuo92, AinsworthOden00,
CockburnKarniadakisShu00, LeVeque02, FischerKruseLoth02}.
This concept is tremendously useful to realize parallel computing, where each
process works on a subregion of the mesh, and their coupling is implemented by
communicating data only between processes that hold adjacent elements.

% Standard FEM, parallel, adaptive
% cite a-posteriori AMR with some history}
% \cite{StrangFix88}
% \cite{AinsworthOden00,BeckerRannacher01}

% semi-L
%\cite{DrakeFosterMichalakesEtAl95}

For some applications, however, nearest-neighbor-only communication is an
undesired constraint.
This applies to element-based particle tracking, such as the
particle/marker-in-cell methods \cite{Harlow64, HarlowWelch65}, used for
example in plasma physics \cite{Dawson83} or viscoelasticity
\cite{MoresiDufourMuhlhaus03}, to semi-Lagrangian methods such as
\cite{DrakeFosterMichalakesEtAl95}, and to smoothed particle hydrodynamics
\cite{GingoldMonaghan77} and molecular dynamics
\cite{EckhardtHeineckeBaderEtAl13}.
Here the mathematical design allows for moving numerical information by more
than one mesh element per time step.
If this is attempted in practice, new ideas are needed to locate points on
non-neighbor remote elements and to find their assigned process.
If the ``points'' are extended geometric objects that can stretch across more
than one element/process, such as in rigid body dynamics
\cite{PreparataShamos85}, an algorithm must cope with multivalued results.
%, which presents an additional challenge.

% TODO: find citation for rigid body interaction
% CB: not now.

% Another instance of what we call non-standard in reference to the
Another generalization of the above-men\-tioned classic methods is the
association of variably sized data to elements.
An obvious example is the $hp$-adaptive finite element method
\cite{SuliSchwabHouston00}, where the data size for an element depends on its
degree of approximation.
More generally, we may think of multiple phases or sub-processes, say physical
or chemical, that differ locally in their data usage.
We may also think of selecting a subset of elements for processing while
ignoring the rest, which can be useful for visualization (visible vs.\
non-visible \cite{FoleyDamFeinerEtAl90}) or file-based output (relevant vs.\
irrelevant according to a user).
Efficiently and adaptively managing such data and repartitioning it between
processes is nontrivial.

% context: HPC, adaptivity
% For uniform meshes and/or small- to medium-sized simulations, providing the
% functionality motivated above is fairly straightforward.
In this paper, we present several new high level algorithms as well as
low level building blocks that perform the operations motivated above.
With ``high level'' we refer to algorithms that appear as a well-defined black
box to the calling code, hiding any kind of mathematical intricacy and
specialty logic on the inside.
%
% further variations thereof.
Our focus is on (a) highly scalable methods that (b) operate on dynamically
adaptive meshes.
Targeting simulations that run on present-day supercomputers, allowing for
meshes that adapt every few, each, or even several times per time step,
requires a carefully designed logical organization of the elements.
To support efficient searches and to aid in creation and partitioning of data,
we choose a combination of a distributed tree hierarchy and a linear ordering
of elements via a space filling curve (SFC) \cite{TuOHallaronGhattas05,
SundarSampathBiros08}.
To allow for general geometries, we generalize to a forest of one or more
trees \cite{StewartEdwards04, BangerthHartmannKanschat07,
BursteddeWilcoxGhattas11}.

% While we build on the \pforest algorithms
% % and its reference
% % software
% \cite{BursteddeWilcoxGhattas11} for our implementation,
% % in practice,
% this is by no means
% required.
% % since we disclose all algorithmic details.% for reimplementation.
% %
% % that we use as starting point.
% %, since they unify all of the above properties.
% %
% % While these concepts are independent to some degree, and some have a
% % decades-old history, we
% %
% % not tied to this implementation:
% We provide sufficient information and background to reimplement the algorithms
% from scratch.

\subsection{Proposed algorithms}

We understand the proposed algorithms as tools that may be relied upon by
conceptually higher-level codes, middleware and libraries, and also by
application developers requiring any of the features offered.
While the following description is abstract and general, the reader is
encouraged to visit Sections \ref{sec:particles} and \ref{sec:spheres} for
concrete, illustrated examples in the context of numerical simulations.

\paragraph{Sparse construction}

The first algorithm, presented in \secref{build}, serves to derive a custom
worker mesh in a direct one-pass algorithm, which reduces the run time over the
common procedure of calling multiple cycles of global refinement and
coarsening.
%
% To solve sub-problems in an end-to-end simulation:
%
This can be useful in material simulations to create an independent forest
adapted to one subsystem (say a fracture zone).
% Another use is to postprocess geology data in a certain subdomain or for a
% certain soil type.
% Last but not least, we may
Another use would be to create a worker forest for just the camera-visible
elements within an in-situ visualization algorithm.
These worker forests can be partitionend independently from the source forest
to run their specific task while preserving overall load-balance.
Our algorithm is sort-free and has sub-second run times throughout.
This improves over bottom-up constructions from scratch
\cite{SundarSampathBiros08}, which have their merit when such heavier
functionality is required by the application, by a large margin.
% ---
%
% Our approach avoids repeated cycles of refinement and coarsening.
%
% Define partition boundary and its encoding
%
% Either change the partition boundary, or locally refine, but not both.

\paragraph{Remote process search}

In \secref{traverse}, we propose a top-down algorithm to find non-local
generalized points in the domain.
These ``points'' may be actual points or particles, extended geometric objects,
or arbitrarily shaped regions in space.
``Non-local'' means that we find the precise intersection of each point with
the set of process subdomains.
% on remote processes.
The algorithm has the following key features.
% \begin{enumerate}
% \item
1.\
We search for multiple points in one pass to amortize mesh memory access,
and we enable multiple match elements/processes per point.
% \item
2.\
We enable both optimistic matching (for example to use fast bounding-box checks
closer to the root) and early discarding (to prune search subtrees as quickly
as possible).
% \item
3.\
We match points to subtrees on remote processes, even though we do not have
access to any remote element (ghost or otherwise).
In fact, the whole algorithm is local and communication-free.
This ostensible paradox is resolved due to our lightweight yet complete
encoding of the forest's partition.
% ---
%
%\end{enumerate}
%
% extended objects, early pruning, optimistic matching
% Non-local element access (ghost layer too small/restrictive)
%
% The search of parallel ownership, for example, is accelerated by the hierarchy
% within each tree, but still loops over the tree roots in the forest, since
% these basically form an unstructured mesh.

\paragraph{Partitioning and parallel I/O}

When writing data to permanent storage, it is an advantage for testing and
general reproducibility if the output format is independent of the number of
processes und the partition of elements that has been used to compute it.
While devising such a format for a single tree is not hard, it becomes
more involved for a forest encoded with minimal metadata.
%, however, such a logic becomes
Another problem related to partitioning is the transfer of fixed- and
variable-sized per-element data, which most applications will allocate in their
own memory space.
For this data, no direct repartitioning interface has been available so far.
We propose two minimal parallel algorithms that perform these tasks in
\secref{partindep}.

% is not available in any straightforward way and
% contributes a whole section to this paper.
%
%%% This feature enables loading and decoding the mesh on arbitrary process counts.
%%% We discuss extensions to save and load not just the mesh, but also fixed-size
%%% and variable-size per-element application data in a partition-independent way.

% We allow for per-element data
% Varying data size per element.
% Save/load this.  Partition this.
%
% \paragraph{Variably sized data}
%
% We touch on
% % several auxiliary communication routines in \secref{auxcomm}.
% %The largest part deals with
% transfering application data on partitioning in \secref{auxcomm}.
% % While many applications use the built-in \pforest feature to transparently
% % repartition a fixed-size per-element payload \cite{Burstedde15}, this does not

% In practice, most applications will allocate per-element data in their own
% memory space, for which no direct repartitioning algorithm is available.
% % allow to hold the data in a linear array, which would often be preferred.
% To lift this limitation, and to extend the functionality to
% % for repartitioning
% variable-size per-element data,
% % as well,
% we include algorithm schematics aligned to our specification of forest
% metadata in \secref{transfer}.

\subsection{Related work}

The need for particle tracking is fundamental in astrophysics and molecular
dynamics simulations.
The use of tree codes for this purpose goes back a long time
 \cite{HernqvistKatz89, McMillanAarseth93}.
It has been found effective to sort the points by interpreting their
tree location relative to a space filling curve
\cite{RahimianLashukVeerapaneniEtAl10, WangZhouShao16}.
% Currently, overall similar program flow as envisaged by us can be found
% in several community codes.
Such functionality is not always custom coded.
The FLASH code, for example, delegates the search of remote particles to the
mesh library \cite{DubeyDaleyZuHoneEtAl12}, while the Chombo library as one
such candidate exposes a formal interface for specifying particle locations
\cite{AdamsColellaGravesEtAl17}.
The Gadget3 code uses a Hilbert curve \cite{RagagninTchipevBaderEtAl16}, and
its feature to form groups of nearby particles for aggregated searching has
been introduced similarly in the molecular dynamics code GROMACS
\cite{AbrahamMurtolaSchulzEtAl15}.

Our algorithms are designed to satisfy the needs of all of the above codes.
In addition, we provide more generality in several respects, where one is that
we (a) allow for searching objects of positive diameter that may each intersect
with more than one process domain, which can be useful in rigid body dynamics
and visualization.
Another important aspect of the present work is that
we do (b) not rely on a single- or multi-width ghost layer or halo region.
Such reliance has frequently been found problematic in all areas from
astrophysics \cite{HowlettManeraPercival15} to smoothed particle hydrodynamics
\cite{WangZhouShao16} and molecular dynamics \cite{GlaserNguyenAndersonEtAl15}.

For efficiency and flexibility, we (c) support grouping in the spirit of
\cite{IsaacBursteddeWilcoxEtAl15}, which offers multiple and optimistic
matching as well as early pruning of search subtrees.
Encouraged by use cases \cite{Albrecht16}, we design our algorithms such that
we (d) require no data transformation \cite{AyachitBauerGeveciEtAl15},
no duplication of data structures
% that has been required previously
\cite{MirzadehGuittetBursteddeEtAl16}, and no parallel sorting
% as used in bottom-up approaches
\cite{SundarSampathBiros08}.
Last but not least, our algorithms enable (e) large scale dynamic adaptive
meshing and repartitioning with sub-second absolute run times.

One of our goals is to avoid communication issues due to \mpialltoall, busy MPI
wait loops, or overallocation of buffers, and we make an effort to calculate
known sender and receiver ranks and message sizes.
%
% Recently, much research has been invested towards efficient
% and space
% filling curves have been found useful in speeding up collective communication
% \cite{LiZhangHoefler18}.
% Here we optimize for the MPI 1.1 standard not only for robustness and
% portability, but also
In doing so, we find that our algorithms are ideally suited to produce the
exact meta information needed for precise point-to-point communication.
This information is similarly fitting for setting up active synchronization in
one-sided MPI \cite{HoeflerDinanThakur14}, such that we support both models.

\subsection{Examples and reproducibility}

\paragraph{Particle tracking example}

% \secref{particles} presents a technology demonstration that executes all
% algorithms put forth in this paper.
In \secref{particles},
we develop an element-based scheme that solves Newton's equations of motion for
a large number of non-interacting particles.
The elements are refined, coarsened, and partitioned dynamically to keep the
number of particles per element near a specified number.
Especially the non-local search of points is crucial to redistribute the
particles to the elements/processes after their positions are updated.

% We find that any algorithm runtimes range between milliseconds and seconds,
% where one second or more only occurs for some algorithms and the largest
% setups.
% All algorithms are found to be practical and scalable to 21e9 particles and
% 64Ki MPI processes.

\paragraph{Remote search example}

In \secref{spheres},
we detail the parallel construction of randomly distributed spheres,
implementing a non-local refinement criterion by the use of the proposed
algorithms.
This approach guarantees a pseudo-random, configurable mesh refinement that
does not depend on the number of processes used to create it.
Thus, we establish and demonstrate partition-independence as an explicit
invariant.

%%%%%%%%%%%%%%%%

% Some of them require only minor changes between
% %, it does not matter if we are
% using one or more trees.

% commpattern

% more applications for later:
% Now describe HP: transfer.
% H P HP \cite{BabuskaGuo92}

% Standard FD %, \cite{MirzadehGuittetBursteddeEtAl16}

% HO/spectral FE/DG methods
% \cite{CockburnKarniadakisShu00}
% \cite{Kopriva09}

% DG adaptive \cite{FischerKruseLoth02}

% H P HP \cite{BabuskaGuo92}

%% PIC
%\cite{Harlow64}
%% MIC
%\cite{HarlowWelch65}
%
%% plasma physics
%\cite{Dawson83}
%
%% Lagrangian large deformation viscoelasticity
%\cite{MoresiDufourMuhlhaus03}

% we do not add these
%% FD marker-in-cell
%\cite{GeryaYuen03}
%% FE particle-in-cell
%\cite{ThielmannMayKaus14}

%%%%%%%%%%%%%%%%

% Hashed octrees:
% N-body \cite{SalmonWarrenWinckelmans94},
% SPH \cite{WarrenSalmon95}.
% Survey \cite{GriebelZumbusch00} ``separators'' as partition markers,
%                                 requires sorting.
% Phd \cite{Gerstner01}, recursive tetrahedral bisection.

%%%%%%%%%%%%%%%%

\paragraph{Software}

% We refer to \pforest as a reference implementation throughout.  Yet, the
% algorithms presented may in principle be implemented in any parallel adaptive
% code that uses linear element storage and provides a shared array of markers
% that encodes the parallel partition.

We provide a reference implementation of all algorithms and example
programs as part of the \pforest software library \cite{Burstedde19}.
\pforest is an MPI-only code, where multiple compute cores per hardware node
are transparently supported by spawning the appropriate amount of MPI ranks.
Explicit shared-memory versions of our algorithms may be written, yet would add
little to the logic we expose in this document---instead, we rely on
optimized shared-memory MPI performance \cite{LiZhangHoefler18}.
We use the prefix \pforestfun{} for functions that can be found in the software
under the same or a similar name, while unprefixed subroutines serve to clarify
the exposition.
We take special care to be explicit about the less obvious mathematical
conventions and tricks that are synergetic to the \pforest design, which will
allow the reader to reimplement all algorithms in their own code.

\paragraph{In-situ visualization}

We refer the interested reader to \cite{Burstedde18b}
for an extended discussion of parallel visualization algorithms based
on the functionality exposed in this paper.

% We describe all algorithms for a multi-tree forest, noting that they can be
% trivially simplified to apply to a single tree.
% % except for pertree

% Applicability for wider range of applications:
% need increased flexibility with respect to multiple aspects.
%
% Variable process counts
% Variable data sizes

% Saving/Loading the mesh on different process counts
% One partition-independent file
% Varying data size per element.
% Save/load this.  Partition this.

% Structure of paper:
%
% 2 conventions
%
% 3 sparse leaves: some details are new, the presentation is
% more of a tutorial that helps to get to grips with the mathematical encoding
% of a parallel forest of octrees.
% Also introduces some subalgorithms required in later sections
%
% 4 search partition
%
% 5 pertree
%   forest-only problem
%
% 6 aux comm
%
% 7 practical example: particle tracking

\section{Principles and conventions}
\seclab{principles}

Our algorithms do not rely on nearest neighbor relations but on the SFC order
that defines and encodes the parallel partition.
As a benefit, the algorithms presented here share the property that the forest
need not be 2:1 balanced and that they do not depend on a ghost layer.
We abstain from incremental tree encodings \cite{BaderBockSchwaigerEtAl10} to
ensure that all elements are individually accessible.
% This information can be encoded globally with very little data,
% which makes it possible to work without an explicit (and necessarily narrow)
% ghost layer.
% This is a more general functionality than implemented in typical parallel PDE
% solvers.
% and allows for a wide
% range of applications in simulation and data management and processing.

While we maintain the notion of elements, they need not necessarily refer to a
classical finite element or a numerical solver context.
% at all.
We allow for arbitrary-size application data to be redistributed in parallel in
the same optimized way that is used for the adaptive mesh, which opens up the
performance and scalability established for managing meshes
\cite{BursteddeGhattasGurnisEtAl10, IsaacBursteddeWilcoxEtAl15}
to many sorts of data.
We hint at various examples and use cases in the respective sections of this
paper.

% this is ok
%\todo{use element until quadrants are introduced}

Throughout the paper, we will be dealing with integers exclusively.
When referring to integer intervals $[ a, b ) \cap \sZ$, we omit the
intersection for brevity.
All arrays are 0-based.
Cumulative arrays (i.e., arrays storing partial sums) are typeset in uppercase
fraktur ($\arrE$).
We denote the number of parallel processes (MPI ranks) by $P$, the number of
trees in the forest of octrees by $K$, and the global number of elements
(leaves of the forest) by $N$.
% We treat them as known constants.
Thus, a process number reads $p \in [ 0, P )$ and a tree number $k \in [ 0, K
)$.

\subsection{Cycles of adaptation}
\seclab{cycles}

In a typical adaptive numerical simulation, the mesh evolves between time steps
in cycles of mesh refinement and coarsening (RC), mesh balancing and/or
smoothing (B), and repartition (P) for load balancing.
Not to be confused with the latter, mesh balancing may refer to establishing a
2:1 size condition between direct neighbor elements
\cite{TuOHallaronGhattas05, SundarSampathBiros08,
BursteddeWilcoxGhattas11, IsaacBursteddeGhattas12} and mesh smoothing to
establishing a graded transition in the sizes of more or less nearby elements.
After \rcbs, the new mesh exists in the same partition boundaries as the
previous one, while families of four (2D) or eight (3D) sibling elements have
been replaced by their parent, or vice versa.
We note that refinement and coarsening is rarely applied recursively,
except for example during the initialization phase of a simulation.

Since \rcbs changes the number of elements independently on each process, load
balance is lost, and P redistributes the elements in parallel to reinstate it.
To guarantee that one cycle of coarsening is always possible, the partition
algorithm may be modified to place every sibling of one family on the same
process \cite{SundarBirosBursteddeEtAl12}.
In some applications it may be beneficial to partition before refinement,
possibly using weights depending on refinement and coarsening indicators, in
order to avoid crashes when one process refines every local element and runs
out of memory.
P is complementary to \rcbs in the sense that it changes the partition boundary
while the elements stay the same.
This design ensures modularity between and flexible combination of individual
algorithms and simplifies the projection and transfer of simulation data
\cite[Figures 3 and 4]{BursteddeGhattasStadlerEtAl08}:

\begin{principle}[Complementarity principle]
\label{principle:complementarity}
A collective mesh operation shall either change the local element sizes within
the existing partition boundary, or change the partition boundary and keep the
elements the same, but not both.
\end{principle}

It should be noted that time stepping is not the only motivation to use
adaptivity:
When utmost accuracy of a single numerical solve is required, we may use
a-posteriori error estimation to refine and solve the same problem repeatedly
at successively higher resolutions;
when setting up a geometric multigrid solver, we create a hierarchy of coarser
versions of a mesh.
In both scenarios, we may add mesh smoothing and most definitely repartitioning
at each level of resolution using the same \rcbsp algorithms, which is key for
the scalability of geometric/algebraic solvers
\cite{BursteddeGhattasStadlerEtAl09,
      SundarBirosBursteddeEtAl12, RudiMalossiIsaacEtAl15}.

\subsection{Encoding a parallel forest}
\seclab{encoding}

We briefly introduce the relevant properties of the \pforest data structures
and algorithms \cite{BursteddeWilcoxGhattas11}, which we see in this paper as a
reference implementation of an abstract forest of octrees.
We consider a forest that is two- or three-dimensional, $d = 2$ or $3$, which
generalizes easily to arbitrary dimensions.
The topology of a forest is defined by its connectivity, \ie, an enumeration of
tree roots viewed as cubes mapped into $\sR^3$ together with a specification of
each one's neighbor trees across the tree faces, edges (in 3D), and corners.
Neighbor relations include the face/edge/corner number as viewed from the
neighbor and a relative orientation, since the coordinate systems of touching
trees need not align.
% As mentioned above, we do not rely on neighbor relations in this paper.

The mesh primitives in \pforest are quadrilaterals in 2D and hexahedra in 3D.
They arise as leaves of a quadtree (2D) or octree (3D), where a root can be
subdivided (refined) into $2^d$ child branches (subtrees).
The subdivision can be performed recursively on the subtrees.
For simplicity, we will
% generally
use the term ``quadrant'' for a tree node.
A quadrant is either a branch quadrant (it has child quadrants) or it is a leaf
quadrant.
The root quadrant is a leaf if the tree is not refined and a branch otherwise.
We call leaves in both 2D and 3D the elements of the adaptive mesh.

In practice, we limit the subdivision to a maximum depth or level $L$, where
the root is at level $\ell = 0$.
Accordingly, a quadrant is uniquely defined by the tree it belongs to, the
coordinates $(x_i) = (x, y, z)$ of its lower left front corner, each an integer
in $\halfopen{0, 2^L}$, and its level $\ell \in \closedin{0, L}$.
A quadrant of level $\ell$ has integer edge length $2^{L - \ell}$, and its
coordinates are integer multiples of this length.
We assume that a space filling curve (SFC) is defined that maps all possible
quadrants of a given level bijectively into an ordered set of curve indices
$\halfopen{0, 2^{d \ell}}$.
We may always embed this index into the space $\halfopen{0, 2^{d L}}$ by
left-shifting by $d(L - \ell)$ bits.
The level may be appended to the curve index to make the index unique across
all levels.

The order defined by the SFC must satisfy a locality property:
The children of a quadrant are called a family and have indices that come after
any predecessor and before any successor of their parent quadrant.
As a consequence, two quadrants are either related as ancestor and descendant,
meaning that the latter is contained in the former, or not intersecting at
all:  Partially overlapping quadrants do not exist.
Common choices of SFC are the Hilbert curve \cite{Hilbert91} and the Morton- or
$z$-curve \cite{Morton66} used in \pforest.
In fact, the algorithms in this paper are equally fit to operate on a
forest of triangles or tetrahedra, as long as its connectivity is well defined
and it is equipped with an SFC such as the one designed for the \tetcode
\cite{BursteddeHolke16}.

% Done
% \todo{impossible that quadrants overlap partially!
% Discuss descendant, ancestor, successor, predecessor somewhere.}

A forest is stored in a data object that exists on each participating process.
Most of its data members are local, that is, apply to just the process where
they are stored, while others are shared, meaning that their values are
identical between all processes.
The shared data is minimal such that it uniquely defines the parallel
partition.
We use linearized tree storage that only stores the leaves and ignores the
non-leaf nodes \cite{SundarSampathBiros08}.
% We also demand that the trees have no gaps, that is, all $2^d$ branches of
% every non-leaf node exist and eventually lead to a full number of leaf nodes.
The leaves are ordered in sequence of the trees, and within each tree in
sequence of the SFC order.
Sometimes we reference local data for the tree with global number $k$ inside a
forest object $s$ by $\Ttree = s.\rtrees[k]$.

The partition of leaves is disjoint, which allows us to speak of the owning
process of an element.
For convenience, the local data of each process includes the numbers of its
first and last non-empty trees.
The trees between and including its first and last are called its local trees.
% An empty tree would be one whose elements are stored entirely on preceding or
% succeeding processes.
The first and last trees of a process may be incomplete, in which case the
remaining elements belong to preceding processes for the first, and to
succeeding processes for the last local tree.
If a process has more than two trees, the middle ones must be complete.
If a process has elements of only one tree, its first and last tree are the
same.
In this case, if that tree is incomplete, its remaining elements may
be on processes both preceding and succeeding.
A process may also be empty, that is, have no elements, in which case it has no
valid first and last tree.

For each of its local trees, a process stores an offset defined by the sum of
local elements over all preceding trees, and the tree's boundaries by way of
its first and last local descendants.
The first (last) local descendant is the first (last) descendant of maximum
level $L$ of its first (last) local element in this tree.
For example, the first local descendant of a complete tree in Morton encoding
has coordinates $x_i = 0$,
while the last has coordinates $x_i = 2^L - 1$, $i \in \halfopen{0, d}$.
The local elements are stored in one flat array for each local tree.
Thus, the tree number for every local element is implicit.
% TODO: replace tree number with tree index?
% CB: No.
Non-local elements are not stored.

% In addition to the connectivity, which is immutable for the lifetime of a
% forest,
The shared data of the forest is the array $\arrE[p]$, the sum of local elements
over all preceding processes, and the array $\arrm[p].(\rtree, \rdesc)$ of the
first local tree and local descendant, for every process $p$.
The first local descendant of a process is identical to the first local
descendant of its first tree.
% and we only store the coordinates of its first corner in the $\rdesc$ field.
Consequently, the array of first descendants is sufficient to recreate the
first and last local descendants $\Ttree.f$, $\Ttree.l$ of any tree local to
any process.
We call $\arrm$ the array of partition markers, since they define the partition
boundary in its entirety (see \figref{pmarkers}).

By design of the SFC, the entries of $\arrm$ are ascending first by tree and then
by the index of the first local descendant.
Whether a process begins
%  on the first corner of
with a given tree and quadrant,
even if the quadrant is non-local and/or a branch, is trivial to check by
examining $\arrm$; see \algref{beginsquadrant} and its use from
\algref{processes} and \algref{treecount}.
%
% \todo{%Include picture of partition markers.
%       Draw process partition, markers, first local descendants}
%
\begin{figure}
\begin{center}
  \includegraphics[width=.9\columnwidth]{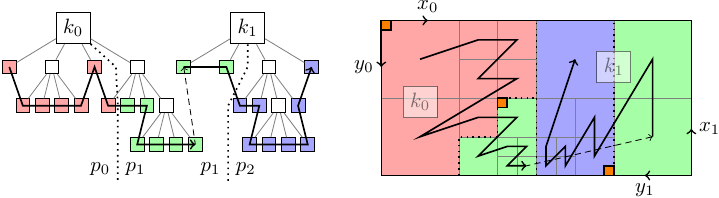}
\end{center}%
\caption{Sketch of a forest of $K = 2$ quadtrees $k_i = i$ (left) and
         the mesh it encodes (right).
         Each tree in the mesh has its own coordinate system that determines
         the order of elements along the space filling curve (black arrows).
         The forest is partitioned between $P = 3$ processes $p_j \equiv j$
         (color coded).
         The partition markers $\arrm[0, 1, 2]$ (orange) are quadrants of
         a fixed maximum level; we do not draw $\arrm[3]$.
         They correspond to the black dotted lines on the left
         that are sometimes called separators \cite{GriebelZumbusch00}.
         This forest is load balanced with cumulative element counts
         $\arrE = [0, 7, 15, 23]$.}%
\label{fig:pmarkers}%
\end{figure}%
\begin{algorithm}
  \caption{\beginsquadrant (process $p$, tree number $k$, quadrant $b$)}
  \alglab{beginsquadrant}
  Determine in $\cO(1)$ whether a hypothetical quadrant in a specific tree is
  the first quadrant on some process $p$.
  This function showcases use of the marker array $\arrm$.
  \newline
  \begin{algorithmic}[1]
    \REQUIRE $b$ is a quadrant in tree $k$
    % \STATE $d \leftarrow$ first descendant of $b$
      \hfill\COMMENT{omit tree ``number'' from now on}
    \RETURN $\arrm[p] = (k, \text{first descendant of $b$})$
    % \ENSURE
    \hfill\COMMENT{comparison yields true or false}
  \end{algorithmic}%
\end{algorithm}%

As stated above, the arrays $\arrE$ and $\arrm$ are available to each process,
% presently requiring 32~bytes per process,
a feature that is crucial throughout.
It has been found exceedingly convenient to store one additional element in
these zero-based arrays, namely $\arrE[P]$ and $\arrm[P]$.
Quite naturally, $\arrE[P]$ is the global number of elements, and the number of
elements on process $p$ is $\arrE[p + 1] - \arrE[p]$ for all $p \in
\halfopen{0, P}$.
Setting $\arrm[P]$
% = (K, (0, 0, 0))$ marks
to the first descendant of the non-existent tree $K$
% This
permits to encode any empty process $p$, including the last one,
by $\arrm[p] = \arrm[p + 1]$, that
is, by successive partition markers being equal in both tree and descendant.
If one or several successive processes are empty, we say that all of them begin
on the same tree and quadrant as the next non-empty process.
By design, \algref{beginsquadrant} returns true for all of them.

It follows from the above conventions that the array $\arrm$ contains
information on the ownership of trees as well:
\begin{property}
  \label{property:gfptreeownership}
    Not every tree needs to occur in $\arrm$.

    % If $k$ is given and process $p$ is the first that satisfies $k =
    % \arrm[p].\rtree$, then the first descendant
    % (e.g., the lower left corner in the Morton encoding) of tree $k$ is
    % in the partition of either $p$ or $p - 1$.
    % More specifically, it is $p$ if and only if $\arrm[p].\rdesc = (0, 0, 0)$.

    % If $k$ occurs and the range of processes $[p, q]$ is widest such that
    % the markers $\arrm[p']$ are equal for $p \le p' \le q$, and $p$ is the
    % first overall that satisfies $k = \arrm[p].\rtree$, then
    % % the processes $p_0$ through $p_1 - 1$ are empty and
    % the first descendant of tree $k$ is in the partition of either $p - 1$ or
    % $q$.
    % More specifically, it is $q$ if and only if
    % $\arrm[p].\rdesc \equiv \arrm[q].\rdesc = (0, 0, 0)$.

    If $k$ occurs and the range of processes $[p, q]$ is widest such that
    \beginsquadrant ($p'$, $k$, $b$) for all $p \le p' \le q$ and the same $b$,
    and $p$ is the first satisfying this condition for any $b$,
    then the first descendant of tree $k$ is in the partition of either $p - 1$
    or $q$.
    More specifically, it is $q$ if and only if
    \beginsquadrant ($q$, $k$, $\rroot$).

    If $k$ does not occur in $\arrm$, then all of its quadrants are owned by
    the last process $p$ that satisfies $\arrm[p].\rtree < k$.
\end{property}
% We will later require to decide whether a process begins on the first corner of
% a quadrant.
% Based on the above, this is a quick check listed in \algref{beginsquadrant}.

%TODO: Discuss what partition markers mean
%
% /* Invariant: Rank p is the first that mentions tree t in gfp[p]
%    and the ownership of t has been assigned to p or p - 1 */
%
% this proc may be empty, so while this is the case skip forward.
%
% Reference \ref{property:gfptreeownership} from algorithm "processes" below.

\section{Forest construction from sparse leaves}
\seclab{build}

In many use cases an application must construct a mesh for which only a small
subset of current elements is relevant:
% Such a case may arise in quite different scenarios:
\begin{itemize}
  \item
    To isolate elements of a given refinement level (and fill the gaps with the
    coarsest possible elements to complete the mesh), for example to implement
    multigrid or local time stepping.
  \item
    % To reduce the size of file output by
    To postprocess only the mesh elements
    selected by a given filter (such as for writing to disk the data of one
    part of a much bigger model).
  \item
    A computation deals with points distributed independently of the element
    partition and varying strongly in density, and we seek to create a mesh
    representing the points.
  \item
    For parallel visualization, we want to process only the part of the mesh
    inside the view angle of a virtual camera.
\end{itemize}
Repeated coarsening addresses only some of these cases and is unnecessarily
slow when it does:
We would execute multiple cycles and carefully maintain data consistency
between them.
Coarsening may also be inadequate entirely, such as in the case of
points where we might want to create a highly refined element for each one,
potentially finer than in the original mesh.

The operation we need is akin to making a copy of the existing mesh (we will
keep the original and its data to continue with the simulation eventually), and
then executing multiple cycles of \rcp on the copy, all in a fast one-pass
design.
In particular, we want to avoid creating forest metadata or element storage and
discarding it again.

While some details in this section are new, the presentation is more of a
tutorial in working with our mathematical encoding of a parallel forest of
octrees.
This section also serves to introduce some subalgorithms required later.

\subsection{Algorithmic concept}
\seclab{buildalgo}

We propose the following procedure \pfbuild:
\begin{enumerate}
  \item
    Initialize an opaque context data structure from an existing source forest
    that will hide temporary working data (\algref{buildbegin}, \buildnew).
  \item
    Add leaves one by one, which need not exist in the source mesh (i.e., they
    may be coarser or finer) but must be contained in the local partition, and
    must be non-overlapping and their index non-decreasing relative to the ones
    added previously.
    These leaves can be sparse (that is, not contiguous in the order of the
    space filling curve; cf.\ \algref{buildadd}, \buildadd).
  \item
    Free the context, not before creating a new forest object as a result: It
    is defined as the coarsest possible forest (a) containing the added leaves
    and (b) respecting the same partition (\algref{buildend}, \buildend).
\end{enumerate}
The resulting forest has the same partition boundary as the source, thus the
above procedure satisfies \principleref{complementarity}.
% By design, the result depends on the partition.
% Should this be undesirable, \rcp cycles can be added, which may still be faster
% than running the cycles without the initial build step.

One advantage is that the construction is process-local, with the caveat
that the result depends on the total number of processes.
However, since the result is a valid forest object, it can be subjected
to calls to \rcbs if so desired, and P in order to load balance it for its
special purpose.
Its number of elements may be smaller than that of the source, possibly by
orders of magnitude, significantly accelerating the computation downstream.

As a difference to
% the bottom-up construction of an octree described in
\cite{SundarSampathBiros08}, we use a source forest to guide the algorithm.
% We let each process work on its own partition without communication.
The monotonicitiy requirement, to add leaves in the order of the source index,
eliminates the linear-logarithmic runtime of a sorting step.
Monotonicity can be realized for example by iterating
through the existing leaves in the source, or by calling the top-down forest
traversal \pfsearch.
% \cite[Algorithm 3.1]{IsaacBursteddeWilcoxEtAl15}.
The latter approach has the advantage that the traversal can be pruned early to
skip tree branches of no interest, not accessing these source elements at all.
% Furthermore, inheriting the source ordering implicitly provides the map between
% a leaf in the source forest that triggers an addition and the leaf added to the
% result, and permits to reprocess the source element data for use with the
% result on the fly.

\pfbuild shares the property of \pfsearch that the serial version is
useful in itself, since the tasks mentioned above may well occur in a
single-process code.
The parallel version of \pfbuild is near identical to the serial one, with
the exception that the local number of leaves in the result, one integer, is
shared with \mpiallgather.
This is standard procedure in \pforest for refinement, coarsening, and 2:1
balance.
Apart from that, the algorithm is communication-free.

% Akin to copy and repeated coarsening, in a fast one-pass implementation.

\subsection{Details description of \pfbuild}
\seclab{builddetails}

\begin{comment}
\begin{algorithm}
  \caption{\begintree (context $c$, tree $k$, element offset $o$)}
  \alglab{begintree}
  We record in the context that we are entering the next tree.
  \newline
  \begin{algorithmic}[1]
    \REQUIRE Source forest $c.s$ and result forest $c.r$ initialized
      (e.g., by \algref{buildbegin})
    \REQUIRE $c.s.\rfirstlocaltree \le k \le c.s.\rlastlocaltree$
    \STATE $c.k \leftarrow k$
    \hfill\COMMENT{remember updated tree number}
    \STATE $c.r.\rtrees[k].o \leftarrow o$
    \hfill\COMMENT{store offset as sum of all preceding elements}
    \STATE $c.\rprev \leftarrow \rinval$
  \end{algorithmic}
\end{algorithm}
\end{comment}
%
\begin{algorithm}
  \caption{\buildnew (source forest $s$) $\rightarrow$ context $c$%
    \hfill (collective call)}
  \alglab{buildbegin}
  This algorithm is the entry point to the sparse construction of a forest.
  We record a reference to the original forest and initialize
  context information to maintain state.
  \newline
  \begin{algorithmic}[1]
    \STATE $c \leftarrow$ new context storing reference to $s$
           and new and empty result forest $r$
    \IF{$s$ has elements on this process}
      \STATE \begintree ($c$, $s.\rfirstlocaltree$, $0$)
        \hfill\COMMENT{set $c.\rprev \leftarrow \rinval$}
    \ENDIF
  \end{algorithmic}
\end{algorithm}
We use a context data structure to track the internal state of building the new
forest from an ascending (and usually sparse) set of local leaves.
It is initialized by \buildnew (\algref{buildbegin}) and contains a copy of the
variables of the source forest that stay the same, most importantly the
boundaries of local trees plus the array of partition markers.
These copies become parts of the result forest at the end of the procedure.
In practice, redundant data may be avoided by copy-on-write.
The state information contains the number of the tree currently being visited
and a copy of the most recently added element, which serves to verify that a
newly added element is of a larger SFC index and not overlapping
(\algref{buildadd}, \lineref{buildaddrequire}).
% and \ref{line:buildaddredundant}).

% No, will not do this
%\todo{Detail members of context structure}

In adding elements, we pass through the local trees in order.
When adding multiple elements to one tree, we cache them and postpone the final
processing of this tree until we see an element added to a higher tree for the
first time.
If at least one element has been added to the tree, we can rely on the
functions \pfcompletesubtree, originally built around a fragment of
CompleteOctree \cite[Algorithm 4, lines 16--19]{SundarSampathBiros08} and
reworked \cite{IsaacBursteddeGhattas12}, and \pfcompleteregion, a
reimplementation of the function CompleteRegion originally described for
\dendro \cite[Algorithm 3]{SundarSampathBiros08}.
Both functions are adapted to the multi-tree data structures of \pforest and
parameterized by the number of the tree to work on.

\begin{algorithm}
  \caption{\pfenlargefirst (quadrant $f$ is modified, quadrant $b$)}
  \alglab{enlargefirst}
  Given a quadrant $f$, we determine whether it is a first child and strictly
  contained in $b$.  If so, we turn it into its parent and repeat.
  This preserves the lower left corner of $f$.
  \newline
  \begin{algorithmic}[1]
    \REQUIRE $f$ is descendant of $b$ (i.e., equal to $b$ or a strict
             descendant of it)
    \STATE $w = f.x \bor f.y \bor f.z$
    \hfill\COMMENT{bitwise or; omit $z$ coordinate in 2D}
    \WHILE{$f.\ell > b.\ell$ \AND ($w \band 2^{L - f.\ell}) = 0$}
      \linelab{firstcomp}
      \STATE $f.\ell \leftarrow f.\ell - 1$
      \hfill\COMMENT{turn $f$ into parent;
                     valid due to $= 0$ comparison in \lineref{firstcomp}}
    \ENDWHILE
    \ENSURE $f$ has the same first descendant as on input
            and is still descendant of $b$
  \end{algorithmic}
\end{algorithm}
\begin{algorithm}
  \caption{\pfenlargelast (quadrant $l$ is modified, quadrant $b$)}
  \alglab{enlargelast}
  This algorithm is the complement to \algref{enlargefirst} in that
  the last (top right back) corner of $l$ is preserved during repeated
  enlargement of the quadrant.
  \newline
  \begin{algorithmic}[1]
    \STATE $\ell \leftarrow l.\ell$; $w = l.x \band l.y \band l.z$
    \hfill\COMMENT{bitwise and; omit $z$ coordinate in 2D}
    \WHILE{$l.\ell > b.\ell$ \AND ($w \band 2^{L - l.\ell}) \ne 0$}
      \STATE $l.\ell \leftarrow l.\ell - 1$
      \hfill\COMMENT{turn $l$ into parent;
                     requires \lineref{lastfix} to become well defined}
    \ENDWHILE
    \STATE $l.x \leftarrow l.x \band \neg (2^{L - l.\ell} - 2^{L - \ell})$
    \hfill\COMMENT{bitwise negation; repeat for $y$ (and $z$ in 3D)}
    \linelab{lastfix}
    \ENSURE $l$ has the same last descendant as on input
            and is still descendant of $b$
  \end{algorithmic}
\end{algorithm}
\begin{algorithm}
  \caption{\ndtree (context $c$) $\rightarrow$ element offset $o$}
  \alglab{endtree}
  In our loop over the local trees of the source forest, we transfer the
  temporary data recorded by adding sparse quadrants
  % (\algref{buildadd})
  into the tree structure of the result forest.
  \newline
  \begin{algorithmic}[1]
    \STATE $\Ttree \leftarrow c.r.\rtrees[c.k]$
    \hfill\COMMENT{reference to result tree data% for convenience
                  }
    \IF{$\Ttree.\relements = \emptyset$}
      \STATE $a \leftarrow$ \pfnearestcommon ($\Ttree.f$, $\Ttree.l$)
      \IF{$\Ttree.f$ is the first descendant of $a$ \AND
          $\Ttree.l$ is its last}
        \STATE $\Ttree.\relements \leftarrow \{ a \}$
        \hfill\COMMENT{tree consists of one element}
      \ELSE
        \STATE $f \leftarrow \Ttree.f$; $l \leftarrow \Ttree.l$
        \hfill\COMMENT{first and last local descendants of tree}
        \STATE $c \leftarrow$ child of $a$ containing $f$;
               \pfenlargefirst ($f$, $c$);
        \linelab{childfc}
        \hfill\COMMENT{modify $f$}
        \STATE $d \leftarrow$ child of $a$ containing $l$;
               \pfenlargelast ($l$, $d$);
        \linelab{childld}
        \hfill\COMMENT{modify $l$}
        \STATE \pfcompleteregion ($\Ttree$, $f$, $l$)
        \hfill\COMMENT{fill elements in $\Ttree$ from $f$ to $l$ inclusive}
      \ENDIF
    \ELSE
      \STATE \pfcompletesubtree ($\Ttree$)
       \hfill\COMMENT{fill gaps % in $\Ttree$
                      with coarsest possible elements}
    \ENDIF
    \RETURN $\Ttree.o + \# \Ttree.\relements$
  \end{algorithmic}
\end{algorithm}
In the event that no element has added to some local tree, we fill the range
between its first and last local descendants with the coarsest possible
elements.
To this end, we first generate the smallest common ancestor of the two
descendants, which contains the local portion of the tree.
If the tree descendants are equal to the ancestor's first and last descendants,
respectively, the ancestor is the tree's only element.
Otherwise, we identify the two (necessarily distinct) children of the ancestor
that contain one of the tree descendants each, and find the descendants'
respective largest possible ancestor that (a) has the same first (last)
descendant and (b) is not larger than the child.
We do this with \algref{enlargefirst} \pfenlargefirst and \algref{enlargelast}
\pfenlargelast, respectively.
We then call \pfcompleteregion with the resulting elements to fill the tree.

The finalization of a tree for the cases discussed above is listed in
\algref{endtree}.
The reader may notice that the logic in Lines~\ref{line:childfc} and
\ref{line:childld}, along with the enlargement algorithms, could be tightened
further by passing just the number $a.\ell + 1$ instead of the children $c$ and
$d$.
We omit such final optimizations in \pforest when not harmful to its
performance, since the information on the child quadrants is valuable for
checking the consistency of the code.
% In practice, we execute the \textbf{Require} and \textbf{Ensure} statements
% that make use of $c$ and $d$ in every debug compile.

\begin{algorithm}
  \caption{\buildadd
           %\newline\mbox{}\hfill
           (context $c$, tree $k$, quadrant $b$, callback \Add)}
  \alglab{buildadd}
  Between \buildnew and \buildend we may add individual sparse quadrants
  that need neither be contiguous nor existing in the source forest.
  \newline
  \begin{algorithmic}[1]
    \REQUIRE $c.k \le k \le c.s.\rlastlocaltree$
    \hfill\COMMENT{adding element to same or higher tree}
    \WHILE{$c.k < k$}
      \STATE $o \leftarrow$ \ndtree ($c$)
      \hfill\COMMENT{finalize current tree, adding its elements to offset}
      \STATE \begintree ($c$, $c.k + 1$, $o$)
      \hfill\COMMENT{commence the next tree}
    \ENDWHILE
    \IF{$c.\rprev \ne \rinval$}
      \REQUIRE $c.\rprev$ less equal and not an ancestor of $b$
      \linelab{buildaddrequire}
      \IF{$c.\rprev = b$}
        \linelab{buildaddredundant}
        \RETURN
        \hfill\COMMENT{convenient exception allows for redundant adding}
      \ENDIF
    \ENDIF
    \STATE $\Ttree.\relements \leftarrow \Ttree.\relements
           \cup \lbrace b \rbrace$
    \hfill\COMMENT{sparse leaves in tree structure until finalized}
    \STATE $c.\rprev \leftarrow b$
      ; \Add ($b$)
      \hfill\COMMENT{optionally initialize application data}
  \end{algorithmic}
\end{algorithm}
We allow to call the \buildadd function repeatedly with the same element, which
is a convenience when using the feature of \pfsearch to maintain a list of
multiple points to search \cite{IsaacBursteddeWilcoxEtAl15}, several of which
may trigger the addition of the current element.
A new element may just as well be finer or coarser than the one in the source,
as long as it is added in order.
The element is added once, and we provide the convenience callback \Add to
establish its application data; see
% The \buildadd function is listed in
\algref{buildadd}.

\begin{algorithm}
  \caption{\buildend (context $c$) $\rightarrow$ result forest $r$%
    \hfill (collective call)}
  \alglab{buildend}
  Tansfer the rest of the temporary context data into the result forest and
  finalize it.
  \newline
  \begin{algorithmic}[1]
    \IF{$c.s$ has elements on this process (else $n \leftarrow 0$)}
      \WHILE{$c.k < c.s.\rlastlocaltree$}
        %\STATE $o \leftarrow$ \ndtree ($c$)
        %\hfill\COMMENT{finalize current tree, adding its elements to offset}
        \STATE \begintree ($c$, $c.k + 1$, \ndtree ($c$))
        %\hfill\COMMENT{commence the next tree}
        \hfill\COMMENT{finalize and commence as above}
      \ENDWHILE
      \STATE local element count $n \leftarrow$ \ndtree ($c$)
      \hfill\COMMENT{we are done with the last local tree}
    \ENDIF
    \STATE $c.r.\rnumbers \leftarrow$ \mpiallgather ($n$)
    \hfill\COMMENT{one integer per process}
    \RETURN $c.r$
    \hfill\COMMENT{also free $c$'s remaining members and $c$ itself}
  \end{algorithmic}
\end{algorithm}

% Collective communication: Allgather local number of elements.

% The source code to this section can be found in the files \pfsearchbuild
% \cite{Burstedde10}.

\section{Recursive partition search}
\seclab{traverse}

Frequently, points or geometrically more complex objects need to be located
relative to a mesh.
The task is to identify one or several elements touching, intersecting, or
otherwise relevant to that object.
There are varied examples of such objects and their uses:
\begin{itemize}
  \item
    Input/output:
    \begin{itemize}
      \item
    Earthquake point sources to feed energy into seismology simulations
      \item
    Sea buoys for measuring the water level in tsunami simulations
    \end{itemize}
  \item
    Numerical/technical:
    \begin{itemize}
      \item
    Particle locations in tracer advection schemes
      \item
    Departure points in a semi-Lagrangian method
    \end{itemize}
  \item
    Geometric shapes:
    \begin{itemize}
      \item
    Randomly distributed grains to construct a porous med\-ium
      \item
    Trapezoids that represent the field of view of a virtual camera
      \item
    Constructive solid geometry objects for rigid body interactions
    \end{itemize}
\end{itemize}
In the following, we refer to all those objects as points.
We distinguish three degrees of generality required depending on the
application.
\begin{enumerate}
  \item
Local:
When it suffices that each process shall identify strictly the points that are
inside its local partition, we may call \pfsearch
\cite[Algorithm 3.1]{IsaacBursteddeWilcoxEtAl15}
to accomplish this task
economically and communication-free.
  \item
Near:
The points are searched in a specified proximity around the local partition.
For example, in most numerical applications we work with direct neighbors in
the mesh.
% Depending on the discretization, this applies to volume, face, edge, and/or
% node data.
Usually we collect one layer of ghost elements that encode the size, position,
and owner process of direct remote neighbors.
% , which we examine to identify the data to be sent and received.
If the ghost elements are ordered by the SFC, they can be searched very much
like the local elements \cite{BursteddeWilcoxGhattas11}.
This principle can be extended to multiple layers of ghosts
\cite
[\pfghostexpand]
{GuittetIsaacBursteddeEtAl16}.
However, the number of ghost layers must be limited, since the number of ghost
elements collected on any given process cannot be much larger than the number
of local elements due to memory constraints.
% A working implementation is found in \pfghostexpand.
  \item
Global:
Every process may potentially ask for the location of every point.
This variant is clearly the most challenging, % and expensive,
since a naive implementation would cause $\cO(P^2)$ work and/or all-to-all
communication.
\end{enumerate}

This section is dedicated to develop a lean and general solution of the global
problem 3.
The main task of the new \algref{searchpartition}, \pfsearchpartition,
%
% We consider two phases as follows:
% First, we
%
% The main task
is to
identify which points match the local partition and which do not,
and in the latter case, which process(es) they match.
% Performing this task efficiently is
% the main result of this section.
% by no means trivial.
It will be advantageous to follow the forest structure top-down to reduce the
number of binary searches and to tighten their ranges as much as possible.
To avoid traversing the forest more than once, we use the top-down context over
all relevant points as a whole.
Given the metadata we hold for the forest, the algorithm is communication-free.

% Second, we can send messages to the processes identified in phase one to
% schedule a local search on the remote process and to receive the results.
% The second phase leads to an asymmetric pattern: Senders know the reciever
% processes, but not the other way around.
% If we use traditional MPI 1.1, which we favor for robustness and portability,
% we will need a subalgorithm to reverse the communication pattern, several of
% which exist; see for example
% \cite{IsaacBursteddeGhattas12, MirzadehGuittetBursteddeEtAl16}
% and \secref{notify}.
% The remainder of the second phase uses standard, possibly asynchronous
% point-to-point MPI messages.

While an all-to-all parallel search is not expected to scale,
our approach is efficient when the application requires data that is
near in a generalized sense but not accessible by Local and Near searches.
% as defined above.
% even though the application requires data that is
% near in a more generalized sense.
If, for example, we search through a neighborhood in space that extends to
a small multiple of the width of a process domain,
such as in a large-CFL Lagrangian method,
we prune the search for the domain outside of the
neighborhood
%, the number of neighbors searched is bounded,
and the procedure scales well.

% We see several advantages in applicability compared to existing approaches that
% employ an SFC-sorting of the points:
% \begin{itemize}
% \item When the points are non-uniformly located in the mesh, we exploit
%       the tree structure to avoid accessing the empty regions.
% \item We allow for a point to produce multiple hits, which is expected
%       to occur when it represents a geometric shape.
% \item We allow for optimistic match queries in following a point through
%       multiple tree branches for early recursions, reducing the number of
%       potentially expensive exact queries.
% \end{itemize}

\subsection{Idea of the recursion}
\seclab{traverseprinciple}

% It remains to specify the first phase of the algorithm, namely the
% identification of owner processes of all points.
We know that the local part of the search can be executed using \pfsearch.
Assuming we remembered all points that do not match locally and run two nested
loops to search each of those points on every remote process, this would be
rather costly.
The alternative of sorting the coordinates of the points in order of the SFC
and comparing it with the partition markers is not applicable when the points
are extended geometric shapes.
Instead, we repurpose the idea behind \pfsearch and apply it to the partition
markers instead of the local quadrants.
This inspires a top-down traversal of the partition of the forest without
accessing any element (which would be impossible anyway, since remote
elements are generally unknown to a given process).

To illustrate the principle, consider a branch quadrant of a given tree and
assume that we know the process that owns its first local descendant and the
one that owns its last.
These two processes define the relevant window onto the array of partition
markers.
Hence, we are done if the first and last process are identical:
This is the owner of all leaves below the branch.
Otherwise, we split the branch quadrant into its $2^d$ children and look for
them in the window of partition markers using a multi-target binary search.
% ;
% see also \secref{traverse}.
This gives us for each child its first and last process, which allows us to
continue this thought recursively, using each child in turn as the current
branch.

The above procedure has several useful properties.
First, to bootstrap the recursion, we execute a loop over all (importantly, not
just the local) trees since a point may exist in any tree.
The partition markers allow us to determine for each tree which processes own
elements of it.
The ascending order of trees, processes, and partition markers inherent in the
SFC allows us to walk through this information quickly.
Furthermore, a leaf can only have one owner process, which means that the
recursion is guaranteed to terminate on a leaf, if not before, even when this
leaf is remote and thus not known to the current process.
% The branch quadrants are constructed without regard to the local elements.

Second, we process all points in one common recursion, which combined with
per-point user decisions of whether it intersects the branch allows us to prune
the search tree early and only follow the relevant points further down.
Both the search window and the set of relevant points shrink
with increasing depth of the branch.
Finally, it is possible to do optimistic matching, meaning returning
matches for a point and more than one branch, which may allow for cheaper
match queries in practice.
Any sharp and more costly matching can be delayed if this is advantageous.
The motivation for this is quite natural in view of searching extended
geometric shapes that may overlap with more than one process partition.
We illustrate the process in \figref{recursion}.
%
% Illustration of recursive idea
\begin{figure}
\begin{center}
  \includegraphics[width=.4\columnwidth]{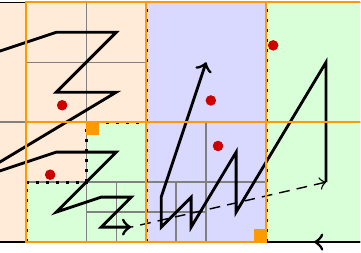}
  \hspace{3ex}
  \includegraphics[width=.4\columnwidth]{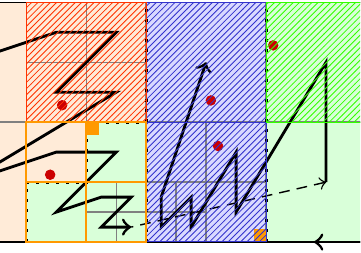}
  \\[2ex]
  \includegraphics[width=.4\columnwidth]{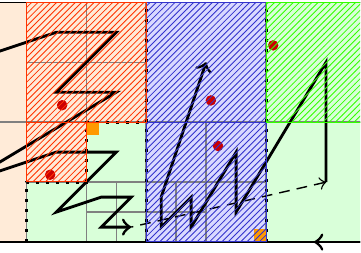}
  \hspace{3ex}
  \includegraphics[width=.4\columnwidth]{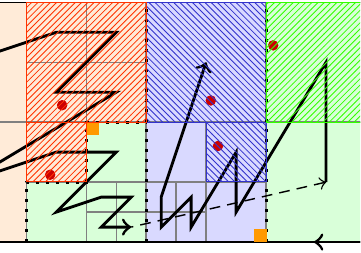}
\end{center}%
\caption{Steps of \pfsearchpartition on process $2$ (blue; cf.\
         \figref{pmarkers}), which must locate 5 points (red), 3 of which
         are not local.
         Top left: recursion at level 1 after 5 search
         quadrants (orange) have returned true in \Match.
         Top right: The red quadrant is remote and the recursion stops
         at the branch containing the point.
         The green quadrant is remote and coincides with a leaf,
         and the recursion stops.
         The two blue quadrants are local, one branch and one leaf,
         and the recursion stops.
         The recursion continues with 4 level 2 quadrants (orange).
         Bottom left: Final result; one search quadrant matches on the red
         process.
         Bottom right: Alternate result obtained by an extension
         % to the algorithm %\`a l\`a
         akin to \pfsearch, making use of the local leaves
         as a convenience (blue hatched backwards).
}%
\label{fig:recursion}%
\end{figure}%

\subsection{Technical description of \pfsearchpartition}
\seclab{traversetechnical}

%% Comment on whether this a documentation of the code, or more or less.

As outlined in \secref{encoding}, \pforest stores one partition marker per
process that contains the number of its first tree.
% We denote the array of markers by $m$.
To find the processes relevant for each tree, we need to reverse this map.
In principle, we could run one binary search per tree to find the smallest
process that owns a part of it.
Instead of doing this and spending $K \log P$ time, we can exploit the
ascending order of both trees and processes, and the fact that the range
of processes for a tree is contiguous, to run the combined and optimized
multi-target search \scsplit presented in \cite{IsaacBursteddeWilcoxEtAl15}.
We restate the precise convention for its input and output parameters in
\algref{scsplit}.
\begin{algorithm}
        \caption{\scsplit (input array $\arra$,
                           offset array $\arrO$,
                           number of types $T$)}
  \alglab{scsplit}
  Interface to multi-objective binary search over a cumulative array
  developed in \cite{IsaacBursteddeWilcoxEtAl15}.
  \newline
  \begin{algorithmic}[1]
    \REQUIRE $\arra$ is sorted ascending by some type $0 \le \arra[i].t < T$
      (repetitions allowed)
    \REQUIRE $\arrO$ has $T + 1$ entries to be computed by this function
    \ENSURE The positions $i$ of $\arra$ that hold entries of type $t$ are
            $\arrO[t] \le i < \arrO[t + 1]$
    \ENSURE If there are no entries of type $t$ in $\arra$, then
            $\arrO[t] = \arrO[t + 1]$
  \end{algorithmic}
\end{algorithm}

% This narrows the size of the window of partition markers examined for a given
% tree in every step of the search:
% If a lower tree and a higher tree own part of a lower and higher process,
% respectively, the range of processes relevant for this tree is between the
% two.
% The result of this procedure is an array of tree markers storing the unique
% lowest process relevant for each tree, in non-decreasing order by construction.

To create the map from tree to process, we use the partition markers $\arrm$ as
input array $\arra$.
We exploit the fact that it has $P + 1$ entries and there is $P'$ minimal such
that $\arrm[p'].\rtree = K$ for all $p' \in [P' , P]$.
Usually, we have $P' = P$, but me way also encounter the case $P' < P$ if the
final range of processes $p \in [P' , P)$ has no elements und hence no trees.
% Even though the tree numbers are zero-based, $0 \le k < K$, there may be
% multiple indices $P' \le p \le P$ with $\arrm[p].\rtree = K$.
Designating the tree number of the partition marker as the type for \scsplit,
we see that we must specify $T = K + 1$ types and the offset array $\arrO$ must
have $K + 2$ entries.
\algref{scsplit} gives us
\begin{equation}
  \arrO[0] = 0
  , \quad
  \arrO[K] = P' \le P
  , \quad \text{and} \quad
  \arrO[K + 1] = P + 1
  .
\end{equation}%
% This covers all possible cases of empty processes (holding no elements and
% thus no trees).

% discuss how beginning/end processes of a tree are determined by this
%
% pnext = p4est_traverse_array_index (tree_offsets, tt + 1);
% P4EST_ASSERT (pfirst <= pnext && pnext <= num_procs);
%
% /* fix the last processor in the tree, which is known at this point */
% P4EST_ASSERT (pnext > 0);
% plast = pnext - 1;
%
Now, running the loop over all trees $0 \le k < K$, we need to determine the
first and last processes $p_\rfirst$, $p_\rlast$ owning elements of tree $k$.
We know for a fact that
\begin{equation}
  p_\rlast = \arrO[k + 1] - 1
  .
\end{equation}%
This can be seen since $p_\rlast \ge \arrO[k + 1]$ would mean that $p_\rlast$
could not have any elements of trees $k$ and less.
And if there were a $p'$ with $p_\rlast < p' < \arrO[k + 1]$, then $p_\rlast$ would
not be the last process of tree $k$.
%
% p4est_traverse_is_clean_start (p4est_t * p4est,
%                                p4est_quadrant_t * quadrant, int p)
% {
%   const p4est_quadrant_t *marker;
%
%   P4EST_ASSERT (p4est != NULL);
%   P4EST_ASSERT (quadrant != NULL);
%   P4EST_ASSERT (0 <= p && p <= p4est->mpisize);
%
%   marker = p4est->global_first_position + p;
%   P4EST_ASSERT (marker->level == P4EST_QMAXLEVEL);
%   P4EST_ASSERT (0 <= marker->p.which_tree &&
%                 marker->p.which_tree <= p4est->connectivity->num_trees);
%   P4EST_ASSERT (marker->p.which_tree == quadrant->p.which_tree);
%
%   return marker->x == quadrant->x && marker->y == quadrant->y
% #ifdef P4_TO_P8
%     && marker->z == quadrant->z
% #endif
%     ;
%
% /* now check multiple cases for the beginning processor */
% if (pfirst < pnext) {
%   /* at least one processor starts in this tree */
%
%   if (p4est_traverse_is_clean_start (p4est, &root, pfirst)) {
%     /* pfirst starts at the tree's first descendant but may be empty */
%     while (p4est_comm_is_empty (p4est, pfirst)) {
%       ++pfirst;
%       P4EST_ASSERT (p4est_traverse_type_tree
%                     (&position_array, pfirst, NULL) == (size_t) tt);
%     }
%   }
%   else {
%     /* there must be exactly one processor before us in this tree */
%     --pfirst;
%     P4EST_ASSERT (p4est_traverse_type_tree
%                   (&position_array, pfirst, NULL) < (size_t) tt);
%   }
% }
% else {
%   /* this whole tree is owned by one processor */
%   pfirst = plast;
% }
%
To determine $p_\rfirst$, we distinguish the cases of (a) no process beginning
in this tree, (b) a process begins at its first descendant, and (c) a process
begins elsewhere in $k$.
We name this algorithm \texttt{processes} (\algref{processes}) and call it
with the the type $t = k$ and the root quadrant of the tree.
\begin{algorithm}
  \caption{\texttt{processes}
          (offset array $\arrO$, type $t$, quadrant $b$)
          $\rightarrow$ ($p_\rfirst$, $p_\rlast$)%
    %\hfill (as part of a loop $k = 0, \ldots, K-1$)
  }
  \alglab{processes}
  Given an offset array of process ranks classified by some type (depending on
  the calling context, for example the first tree of a process or the first child
  index relative to a search quadrant), determine the tightest inclusive range
  of processes of this type.
  \newline
  \begin{algorithmic}[1]
    \REQUIRE By context, $b$ is a quadrant in some tree $k$
    \STATE $p_\rlast \leftarrow \arrO[t + 1] - 1$
      \hfill\COMMENT{this value is final}
    \STATE $p_\rfirst \leftarrow \arrO[t]$
      \hfill\COMMENT{initialization}
  %  \IF{$p_\rfirst > p_\rlast$}
  %    \STATE $p_\rfirst \leftarrow p_\rlast$
  %    \hfill\COMMENT{no new process begins with $t$}
  %  %\IF{$p_\rfirst < \arrO[t + 1]$}
  %    %\ELSIF{$p_\rfirst$ begins on the first corner of $b$}
      \IF{$p_\rfirst \le p_\rlast$
          \AND
            % $p_\rfirst$ begins on the first corner of $b$
            \beginsquadrant ($p_\rfirst$, $k$, $b$)}
        \WHILE{$p_\rfirst$ is empty}
          \STATE
            % $p_\rfirst \leftarrow p_\rfirst + 1$
            $p_\rfirst \pplus$
            \hfill\COMMENT{empty processes use same type as their successor}
        \ENDWHILE
      \ELSE
        \STATE
          % $p_\rfirst \leftarrow p_\rfirst - 1$
          $p_\rfirst \mminus$
          \hfill\COMMENT{there must be exactly one
                         earlier process for this type}
      \ENDIF
    %\ELSE
    %  \STATE $p_\rfirst \leftarrow p_\rlast$
    %  \hfill\COMMENT{no new process begins with $t$}
%   %                    it occurs for a single process}
    %\ENDIF
    \ENSURE Range $[p_\rfirst, p_\rlast]$ is widest s.t.\ each end has at
            least one item of type $t$
  \end{algorithmic}
\end{algorithm}

We show the toplevel call \pfsearchpartition in \algref{searchpartition}.
For clarity, we have excluded the local search of points (covered in detail in
\cite{IsaacBursteddeWilcoxEtAl15}) and reduced the presentation to the search
over the parallel partition.
Since it does not communicate, it can be called by any process at any time.
It identifies the relevant processes for each tree in turn as discussed above
and then invokes the recursion for each tree.
The recursion keeps track of the points to be searched by a user-defined
callback function \Match.
This callback is passed the range of processes relevant for the current branch
quadrant and may return false to indicate an early termination of the
recursion.
The points and the callback to query them do not need to relate to invocations
on other processes.
\begin{algorithm}
\caption{\pfsearchpartition
        (point set $Q$,
         callback \Match)%
  }
\alglab{searchpartition}
Toplevel interface for generalized partition search.
We use the multi-objective binary search to bootstrap the range of
relevant trees, and we make sure that the range of processes per tree
is tight (Algorithms \ref{alg:scsplit} and \ref{alg:processes}) before
entering the per-tree recursion.
\newline
\begin{algorithmic}[1]
\STATE \scsplit ($\arrm$, $\arrO$, $K + 1$)
\hfill\COMMENT{split partition markers $\arrm$ by their tree number}
\FORALL{tree numbers $0 \le k < K$}
  \STATE % $r \leftarrow \rtrees[k].\rroot$
         $a \leftarrow \rroot$
  % \hfill\COMMENT{quadrant at level $0$, need not be local}
  \hfill\COMMENT{construct toplevel quadrant to begin}
  \STATE ($p_\rfirst$, $p_\rlast$) $\leftarrow$
         \texttt{processes} ($\arrO$, $k$, $a$)
  \hfill\COMMENT{potential owners of quadrants in $k$}
  \STATE \texttt{recursion} ($a$, $p_\rfirst$, $p_\rlast$, $Q$, \Match)
  \hfill\COMMENT{bootstrap recursion for tree $k$}
\ENDFOR
\end{algorithmic}
\end{algorithm}

% /* find the processors for all children of the quadrant */
% sc_array_init_view (&pview, rec->position_array,
%                     pfirst + 1, plast - pfirst);
% sc_array_init_size (&offsets, sizeof (size_t), P4EST_CHILDREN + 1);
% sc_array_split (&pview, &offsets, P4EST_CHILDREN,
%                 p4est_traverse_type_childid, quadrant);
% P4EST_ASSERT (offsets.elem_count == (size_t) (P4EST_CHILDREN + 1));
% P4EST_ASSERT (p4est_traverse_array_index
%               (&offsets, P4EST_CHILDREN) == (size_t) (plast - pfirst));
% P4EST_ASSERT (p4est_traverse_array_index (&offsets, 0) == 0);
%
% Knowing which processes own parts of tree $k$, we begin the recursion at the
% tree root;
The recursion is detailed in \algref{searchrecursion}.
Each step takes a branch quadrant $b$ and the first and last processes that own
elements of it.
% , denoted by $q_\rfirst$ and $q_\rlast$.
If they are the same, this is the owner of all elements below $b$ and the
recursion ends.
% It is crucial to call for all points before terminating the recursion in order
% to process all points at the leaf level.
Otherwise,
% $q_\rfirst < q_\rlast$,
the task is to find the first and last processes $p_{i,\rfirst}$ and
$p_{i,\rlast}$ for each child $c_i$ of $b$.
Here we use \scsplit with an input array that is the minimal window on the
markers, defined by
\begin{equation}
  \eqnlab{Adefinerecursion}%
  \arra[j] = \arrm[p_\rfirst + 1 + j]
  \quad\text{for}\quad
  0 \le j < \Delta p = p_\rlast - p_\rfirst
  .
\end{equation}%
This ensures that all elements of $\arra$ refer to processes beginning inside $b$.
We set their type to the number of the child of $b$ in which they begin, which
fixes $T = 2^d$ and yields
\begin{equation}
  \arrO[0] \ge 0
  ,
  \quad
  \arrO[2^d] = \Delta p
  ,
  \quad\text{and}\quad
  p_{i,\rlast} = \arrO[i + 1] + p_\rfirst
  .
\end{equation}%
If we want to repurpose
% \algref{processes}
\texttt{processes}
to determine $p_{i,\rfirst}$ and
$p_{i,\rlast}$, we need to make sure that the offset array indexes into processes,
which we accomplish by adding $p_\rfirst + 1$ to each of its elements
(\lineref{Opluspfirst}% of \algref{searchrecursion}
) to correct for the window selection
\eqnref{Adefinerecursion}.
\begin{algorithm}
\caption{\texttt{recursion}
        \newline\mbox{}\hfill
        (quadrant $b$,
         processes $p_\rfirst$, $p_\rlast$,
         point set $Q$,
         callback \Match)%
  }
\alglab{searchrecursion}
Recursion step of the partition search.
It is allowed for the callback to match without knowing whether this
match is final.
Note that we traverse each subtree at most once.
\newline
\begin{algorithmic}[1]
\REQUIRE By context, $b$ is a quadrant in some tree $k$
\REQUIRE The first descendant of $b$ is owned by $p_\rfirst$,
         its last by $p_\rlast$
\STATE Set of matched points $M \leftarrow \emptyset$
\FORALL{$q \in Q$}
  \IF{\Match ($b$, $p_\rfirst$, $p_\rlast$, $q$)}
    \STATE $M \leftarrow M \cup \{ q \}$
    \hfill\COMMENT{determine the points that we keep}
  \ENDIF
\ENDFOR
\IF{$M = \emptyset$ \OR $p_\rfirst = p_\rlast$}
  \RETURN since all matches failed and/or
          all quadrants below $b$ belong to $p_\rfirst$
\ENDIF
% \STATE $C \leftarrow$ \texttt{Children} ($b$)
\STATE \scsplit ($\arrm[p_\rfirst + 1, \ldots, p_\rlast]$, $\arrO$, $2^d$)
\hfill\COMMENT{split by child id relative to $b$}
\FORALL{$c_i \in \texttt{Children}$ ($b$), $0 \le i < 2^d$}
  \STATE ($p_{i,\rfirst}$, $p_{i,\rlast}$) $\leftarrow$
         \texttt{processes} ($\arrO + p_\rfirst + 1$, $i$, $c_i$)
  \hfill\COMMENT{owning descendants of $c_i$}
  \linelab{Opluspfirst}
  \STATE \texttt{recursion} ($c_i$, $p_{i,\rfirst}$, $p_{i,\rlast}$, $M$, \Match)
  \hfill\COMMENT{pursue remaining points to bottom}
\ENDFOR
\end{algorithmic}
\end{algorithm}

%\section{Auxiliary communication routines}
%\section{Partition-independent linearization}
\section{Partitioning and parallel I/O}
\seclab{partindep}

This section introduces parallel algorithms that support
partition-independent storage of the mesh and the communication of application
data between an old and new parallel partition.
One guideline that we follow throughout is the following.
\begin{principle}[partition independence]
  \label{principle:arbitrarympisize}
On writing, the organization and contents of file(s) written for a given state
of data shall be independent of the parallel partition of the simulation.
On reading, any number of processes shall be suitable to read such a file
(provided that the total memory available is sufficient).
\end{principle}

Partition independence is a valuable idea for a multitude of reasons, such as:
\begin{enumerate}
\item
Data is often transferred to a different computer for post-processing, having a
different number of processors and a different runtime/batch system.
\item
The scalability of post-processing algorithms is usually less than that of
simulation algorithms.
\item
We would like to make regression-testing, reproduction and post-processing least
restrictive and most convenient further down the data processing chain.
\end{enumerate}

% The forest of linear octrees constitutes a basis well suited for the above
% functionality.
% We will introduce specific requirements and recommended procedures in the
% remainder of this Section.

\subsection{Writing element counts per tree}
\seclab{pertree}

% Our algorithm efficiently obtains a count of elements per tree across process
% boundaries.
% Its technical design is presented in \secref{pertree}.
% Since the need for this algorithm is not obvious at first, we begin with a more
% general discussion.

% It is frequently required to save the state of a simulation to a long-term
% storage medium, say to a file on a hard disk system.

% It is a frequent observation that file I/O
% can be slower than the simulation itself, and sometimes entirely
% impractical, due to physical limits and the data size that needs to be written.

% We recognize that writing the state, either raw for small- to medium-size
% simulations or compressed by adaptation, will remain an indispensable operation
% in practice even when in-situ post-processing has become widely available.

The element counts per tree are a partition-invariant property of a global
forest mesh and thus important to define a concise and complete mesh I/O format.
However, they are not maintained within the internal state of the parallel
forest, which prompts us to develop a dedicated algorithm to compute them using
minimal time and communication.

% Let us thus examine how to save the state in the most concise, flexible, and
% general way.

% In reality, we are not yet at a point that alternatives to writing files are
% widely available and used in practice, even though considerable research exists
% on the subject.

% To make the most out of the predominant simulation work flow, we examine how to
% save the state in the most concise, flexible, and general way.

% When using non-uniform meshes, writing numerical data must be accompanied by
% writing the mesh, otherwise it will not be possible to load and recreate the
% information necessary to analyze or continue the simulation in the future.
% The least common denominator to create a partition-independent state is then to
% write one file for the mesh and one file for each set of per-element numerical
% data.

Let us first consider the (simpler) situation of a one-tree forest.
If we were to include $P$ and the arrays $\arrm$ and $\arrE$ in the mesh file,
it would not be partition independent.
Thus, the only header information permitted is the global element count $N =
\arrE[P]$.
In practice, it is written by the first process, but any other process would be
able to write the header as well.
For each element we store its coordinates $x_i$ and the level, which are of
fixed size $s$.
The window of the mesh file to be written by process $p$ is
\begin{equation}
  \eqnlab{filewindow}
  \text{size of header }+ s \times \halfopen{ \arrE[p], \arrE[p + 1]}
  ,
\end{equation}
which is easily done in parallel using the MPI I/O standard.
On reading, each process learns the values $p$ and $P$ from the MPI environment
and reads the header to learn $N$.
This is sufficient to compute a new array $\arrE$ \cite[equation
(2.5)]{BursteddeWilcoxGhattas11}, which is in turn sufficient to read the local
elements from the file by \eqnref{filewindow}.
The first element read fixes the local partition marker $\arrm[p]$, while
an empty process sets it to an invalid state.
The partition markers are shared by one call to \mpiallgather and examined once
to repair the invalid entries due to empty processes.

For a multi-tree forest, we encounter two additional tasks.
The first is writing the number of trees and their connectivity to the file,
for which we exploit the fact that the connectivity is known to each process in
the current \pforest design.
The second task is deeper:
When reading the window of local elements, it is not known which tree(s) they
belong to.
Of course, we may store the tree number in each element, but this would be
redundant and add some dozen percent to the file size.
One way to encode the tree assignment of elements efficiently is to
postulate an array $\arrN$ of cumulative global element counts over trees and
to include it in the header.

Let $N_k > 0$ be the global number of elements in tree $k$ that is generally
not available from the distributed data structure.
Our goal is to compute these counts and encode them in a cumulative array $\arrN$
with $K + 1$ non-decreasing entries,
\begin{subequations}
\begin{gather}
  \eqnlab{cumulativeN}.
  \arrN[k'] = \sum_{k = 0}^{< k'} N_k
  ,
  \quad 0 \le k' \le K
  \qquad\Rightarrow
  \\
  \arrN[0] = 0
  ,
  \qquad
  \arrN[k + 1] - \arrN[k] = N_k
  \:,
  \qquad
  \arrN[K] = \sum_{k = 0}^{< K} N_k = N
  .
\end{gather}%
\end{subequations}%
This format is convenient in facilitating binary searches through the results.
Note that any $N_k$ may be greater equal $2^{32}$ and thus requires 64 bits of
storage.
% holding cumulative instead of per-tree counts does not change the memory
% footprint.

We aim to avoid the communication and computation cost $\cO( K P )$
of a naive implementation, i.e., one that has every process count the elements
in every tree.
Our proposal is to define a unique process responsible for computing the
element count in any given tree,
% let the element count of each tree be collected by a given process
% We propose an algorithm that
and to minimize communication by sending at most one
message per process to obtain the counts.
This shall hold even if a process is responsible for more than one tree.
Multiple conventions are thinkable to decide on the responsible process, where
we demand that the decision is made without communication.
We also demand that all pairs of sender and receiver processes are decided
without communication.
One suitable choice is the following.
\begin{convention}
  \label{convention:treecounter}
  The process $p$ responsible for computing the number of elements in tree $k$,
  which we denote by $p_k$, is the one that owns the first element in $k$,
  unless more than one process has the first descendant of tree $k$ as their
  partition marker.
  In the latter case, we take $p_k$ as the first process in that set, which is
  necessarily empty.
\end{convention}

This convention ensures that the range of trees that a process is responsible
for is contiguous (or empty).
In addition, it guarantees that $k < k'$ implies $p_k \le p_{k'}$.
Allowing for empty processes to be responsible fixes $p_0 = 0$ in all cases.
\begin{property}
  \label{property:emptyonetree}
  An empty process is responsible for at most one tree.
\end{property}
\begin{proof}
  If an empty process were responsible for two different trees, both would have
  to occur in its partition marker, which is impossible by definition.
\end{proof}%\newline

Let us proceed by listing the phases of the algorithm $\arrN \leftarrow
\pfpertree$.
\begin{enumerate}
\item
  Determine for each process $p$ the number of trees
  % $K_p$
  that it is
  responsible for,
  \begin{equation}
    K_p = \# \{ k : p_k = p \}
    , \qquad
    % \qquad\Rightarrow\qquad
    0 \le K_p \le K
    .
  \end{equation}%
  We may additionally define an array $\arrK$ of cumulative counts,
  % We denote it by $M$, and it has $P$ elements satisfying
  % Thus,
  \begin{equation}
    \eqnlab{treecountsoffsets}
    % 0 \le K_p \le K
    % ,
    % \qquad
    \arrK[p'] = \sum_{p = 0}^{< p'} K_p
    \qquad\Rightarrow\qquad
    % \:,
    % \qquad
    \arrK[0] = 0
    , \qquad
    \arrK[P] =
    % \sum_{p = 0}^{< P} K_p =
    K
    .
  \end{equation}%
  Due to the design of the partition markers and \conventionref{treecounter},
  every process populates these arrays identically in $\cO(\max\{ K, P \})$
  time, requiring no communication.
  We provide \algref{treecount} to detail this computation.
\begin{algorithm}
  %\caption{We count the number of trees any process is responsible for.}
  \caption{\texttt{responsible} (computes arrays of tree counts $(K_p)$,
                                 tree offsets $\arrK$)}
  \alglab{treecount}
  Compute which processes are responsible for counting elements in which trees.
  This is a process-local preparation for establishing the global per-tree
  element counts.
  \newline
  \begin{algorithmic}[1]
    \STATE
           $p \leftarrow 0$;
           $k \leftarrow 0$;
           $K_0 \leftarrow 0$
      % \hfill\COMMENT{$K_0$ is not final}
      % \hfill\COMMENT{initialize}
    \LOOP
      \ENSURE Process $p$ is the minimum of all $p'$ with
              $k = \arrm[p'].\rtree$ (cf.\ \propertyref{gfptreeownership})
      \ENSURE Responsibility for $k$ has been assigned to either $p$ or $p - 1$
      \REPEAT
        \STATE
           % $p \leftarrow p + 1$;
           $p \pplus$;
           $K_p \leftarrow 0$
        \hfill\COMMENT{find the first process that begins in a later tree}
      \UNTIL{$\arrm[p].\rtree > k$}
      \STATE
           % $k \leftarrow k + 1$
           $k \pplus$
        \hfill\COMMENT{proceed to that tree incrementally}
      \WHILE{$k < \arrm[p].\rtree$}
        \STATE
           % $K_{p - 1} \leftarrow K_{p - 1} + 1$;
           $K_{p - 1} \pplus$;
           $k \pplus$
          \hfill\COMMENT{while assigning in-between trees}
      \ENDWHILE
      \IF{$k = K$}
        \STATE
          $K_{p'} \leftarrow 0$ \textbf{forall} $p' \in [p + 1, P)$;
          \textbf{break loop}
          \hfill\COMMENT{assign remaining slots}
      %\IF{$k < K$}
        \ELSIF{% $\arrm[p].\rdesc = (0, 0, 0)$
               \beginsquadrant ($p$, $k$, $\rroot$)%
        }
              %\hfill\COMMENT{begins on first corner of tree}
          % \STATE $p' = p$
          \STATE $K_{p} \pplus$
            \hfill\COMMENT{it is legal if $p$ is empty}
        \ELSE
          % \STATE $p' = p - 1$
          \STATE $K_{p - 1} \pplus$
            \hfill\COMMENT{$p - 1$ is never empty}
        \ENDIF
      % \ELSE
      %   \STATE
      %      $K_{p'} \leftarrow 0$ \textbf{forall} $p' \in [p + 1, P)$;
      %   % \WHILE{$p < P$}
      %   %   \STATE
      %   %    $p \leftarrow p + 1$;
      %   %    $K_p \leftarrow 0$
      %   % \ENDWHILE
      %   % \STATE
      %            \textbf{break loop}
      % \ENDIF
%
%   for (;;) {
%     /* Invariant: Rank p is the first that mentions tree t in gfp[p]
%        and the ownership of t has been assigned to p or p - 1 */
%     P4EST_ASSERT (gfp[p].p.which_tree == t);
%     P4EST_ASSERT (p == 0 || gfp[p - 1].p.which_tree < t);
%     do {
%       treecount[++p] = 0;
%     }
%     while (gfp[p].p.which_tree == t);
%     /* Assign the trees before the next first quadrant */
%     P4EST_ASSERT (t < gfp[p].p.which_tree);
%     for (++t; t < gfp[p].p.which_tree; ++t) {
%       ++treecount[p - 1];
%     }
%     if (t < num_trees) {
%       P4EST_ASSERT (p < num_procs);
%       /* Check if the processor has the beginning of the tree */
%       if (gfp[p].x == 0 && gfp[p].y == 0
% #ifdef P4_TO_P8
%           && gfp[p].z == 0
% #endif
%         ) {
%         ++treecount[p];
%       }
%       else {
%         ++treecount[p - 1];
%       }
%     }
%     else {
%       while (p < num_procs) {
%         treecount[++p] = 0;
%       }
%       break;
%     }
%   }
%
    \ENDLOOP
    \STATE Compute
           % tree offsets from counts
           $\arrK$ from $(K_p)$ by \eqnref{treecountsoffsets}
  \end{algorithmic}
\end{algorithm}
\item
  While the previous step is identical on all processes, let us now take the
  perspective of an individual process $p$ with $K_p > 0$.
  % that is responsible
  % for counting at least one tree.
  It must obtain the number of elements in each of the $K_p$ trees it is
  responsible for and store the result, say, in an array $\arrn$ of the same
  length.
  We initialize each slot with the number of process-local elements in that
  tree,
  \begin{equation}
    \Ttree_i = \rtrees \left\lbrack \arrK[p] + i \right\rbrack
    , \quad
    \arrn [i] = \# \Ttree_i.\relements
    , \quad
    \text{for all $i \in \halfopen{ 0, K_p }$.}
  \end{equation}%
  % We observe the following.
  % Due to \conventionref{treecounter}, we observe the following:
  \begin{proposition}
    \label{proposition:allbutfinal}
    The counts in all but the last element of $\arrn$ are final,
    \begin{equation}
      \arrn[i] = N_k
      \quad\text{for all $k - \arrK[p] = i \in \halfopen{ 0, K_p - 1 }$.}
    \end{equation}%
  \end{proposition}%
  \begin{proof}%
    If process $p$ is empty, it is responsible for at most one tree, $K_p \le
    1$, so there is nothing to prove.
    Otherwise, it owns the first element of every tree it is responsible for.
    This means that all but the last one of these trees are complete on $p$ and
    their number of local elements is also their global number of elements.
  \end{proof}
\item
    %\mpiirecv
  \label{itemmpiirecv}
  It remains to determine the number of remote elements in the last tree
  $k = \arrK[p + 1] - 1$ that a process is responsible for.
  They are necessarily located on higher processes.
  First, we add the elements of the subsequent processes that begin and end
  in this same tree.
  Identifying these processes is best expressed as a C-style code snippet:
  % (\algref{incempties}).
  \begin{equation}
  % \begin{algorithm}
  %   \caption{$q \leftarrow$ \texttt{advance\_same\_tree}}
    \eqnlab{incempties}
  %   \begin{algorithmic}
  %     \STATE
    \text{\textbf{for}
          ($q \leftarrow p + 1$;
           % $q < P \wedge K_q = 0$;
           $q < P$ \textbf{and} $K_q = 0$;
           $q \pplus$)
          \texttt{\{\}}}
  %   \end{algorithmic}
  % \end{algorithm}
  \end{equation}
  The addition itself is quick by using the cumulative element counts,
  \begin{equation}
    \arrn_\Delta = \arrE[q] - \arrE[p + 1]
    ,
  \end{equation}
  where we benefit from the convention that $\arrE[P] = N$.
  If the process $q$ that the loop \eqnref{incempties} ends with begins on the
  next highest tree, it does not contribute elements to $k$, and we set
  $\arrn_q = 0$.
  This condition applies as well if there are no more processes, $q = P$, due
  to the definition of $\arrm[P]$.
  Otherwise, $k$ is $q$'s first local tree, and we require $q$ to send a message
  that contains its local count of elements in this tree, which $p$ receives as
  $\arrn_q$.
  Either way, the final element count is obtained by the update
  \begin{equation}
    \arrn[K_p - 1] \leftarrow \arrn[K_p - 1] + \arrn_\Delta + \arrn_q
    .
  \end{equation}
%      if (gfp[p].p.which_tree == t) {
%        P4EST_ASSERT (p < num_procs);
%        /* Processor p has part of this tree too and needs to tell me */
%        mpiret = sc_MPI_Irecv (&recvbuf, 1, P4EST_MPI_LOCIDX, p,
%                               P4EST_COMM_COUNT_PERTREE, p4est->mpicomm,
%                               &req_recv);
%        SC_CHECK_MPI (mpiret);
%        addtomytree = c;
%      }
%      else {
%        P4EST_ASSERT (p <= num_procs);
%        P4EST_ASSERT (gfp[p].p.which_tree == t + 1);
%      }
\item
    %\mpiisend
  We have seen above that some processes are required to send a message
  containing the count of local elements in their first local tree to a lower
  process.
  By the reasoning in~\ref{itemmpiirecv}., the processes that send a message
  are precisely those that are responsible for at least one tree and own at
  least one element in a preceding tree.
  The condition for process $p$ being a sender is thus
  \begin{equation}
    \eqnlab{conditionsend}
    K_p > 0
    \quad \wedge \quad
    \arrm[p].\rtree < \arrK[p]
    .
  \end{equation}
  What is the receiving process?  Again, the answer is a short loop:
  \begin{equation}
    \eqnlab{decempties}
    \text{\textbf{for}
          ($q \leftarrow p - 1$;
           $K_q = 0$;
           $q \mminus$)
         \texttt{\{\}}}
  \end{equation}
%   if (mycount > 0 && (t = gfp[rank].p.which_tree) < treeoffset[rank]) {
%     /* Send information to processor that counts my first local quadrants */
%     P4EST_ASSERT (rank > 0 && p4est->first_local_tree == t);
%     tree = p4est_tree_array_index (p4est->trees, t);
%     /* Always below the 32bit limit since this is processor-local data */
%     sendbuf = (p4est_locidx_t) tree->quadrants.elem_count;
%     for (p = rank - 1; treecount[p] == 0; --p) {
%       P4EST_ASSERT (p > 0);
%     }
%     mpiret = sc_MPI_Isend (&sendbuf, 1, P4EST_MPI_LOCIDX, p,
%                            P4EST_COMM_COUNT_PERTREE, p4est->mpicomm,
%                            &req_send);
%     SC_CHECK_MPI (mpiret);
%   }
  \begin{property}
    \label{property:underrun}
    It is guaranteed that the loop does not underrun $q = 0$.
  \end{property}
  \begin{proof}
    The initialization is safe due to \eqnref{conditionsend}, which implies that
    a sender always satisfies $p > 0$.
    Furthermore, if all preceding processes had $K_{p'} = 0$, then $p$ would be
    responsible for tree $k = 0$, which would contradict \eqnref{conditionsend}.
  \end{proof}
\item
  \label{itemallgatherv}
  At this point, every process has computed $\arrn$, the global count of
  elements in every tree that it is responsible for.
  If such distributed knowledge suffices for the application, we may stop here.
  If it should be shared instead, we can reuse the arrays $(K_p)$ and $\arrK$
  to feed one call to \mpiallgatherv (they have the correct format by design).
  The amount of data gathered is one long integer per tree, thus the total data
  size is $K$ times 8 bytes.
  % and the runtime is determined purely by the latency of the call.
\end{enumerate}
Computing the cumulative counts $\arrN$ from the freshly established values
$N_k$ is straightforward by \eqnref{cumulativeN}, assuming that the final phase
\ref{itemallgatherv} is executed to share $(N_k)$ between all processes.
The algorithm \pfpertree does work of the order $\cO (\max \{ K, P \})$, where
the constant is negligible since the computations are rather minimalistic.
What is more important is that we send strictly less than $\min \{ K, P \}$
point-to-point messages between known ranks, all of them carrying one integer,
and each process being sender and/or receiver of at most one message.
We expect such a communication to be fast.
% require negligible time.

% Determining the number of elements per tree may be done in several phases.

Going back to our original motivation to store and load partition-independent
forest files, may may add that, mathematically speaking, we could skip phase 5
and delegate the writing of $\arrN$ to parallel MPI I/O.
In practice, however, it is simpler and quite probably quicker to execute phase
5 and have rank zero write all of $\arrN$ into the file header, since it writes
the rest of the header anyway.

\subsection{Data transfer on repartitioning}
\seclab{transfer}

Like all \pforest algorithms, \pfbuild and \pfsearchpartition are agnostic of
the application.
They provide callbacks \Add and \Match as a convenient way for the application
to access and modify per-element data.
By its original design, the \pforest implementation manages a per-element
payload of user-defined size, which is convenient for storing flags or other
application metadata.
This data is preserved during \rcbs for elements that do not change,
and may be reprocessed by callbacks for elements that do.
The data is sent and received transparently during partition P, which means
that it persists throughout the simulation.
However, we do not recommend to store numerical data via the payload mechanism,
since this memory is expected to fragment progressively by adaptation.
It will be more cache efficient to allocate a contiguous block of memory that
is accessed in sequence of the local elements
\cite{BursteddeBurtscherGhattasEtAl09}, either as an array of structures or as
multiple arrays.
Such memory is allocated in application space, and so far there is no general
function to transfer it when the forest is partitioned.
In the following, we
% present \pftransferfixed and \pftransfervariable to
outline algorithms to accomplish this for fixed and variable per-element
data sizes, respectively.

% \subsubsection{Design of the call}

As described in \secref{encoding}, the partition of the forest is stored by the
markers $\arrm$ and the local element counts $\arrE$.
If we consider a forest before and after partition (an operation that adheres
to \principleref{complementarity}), the only difference between the two forests
is in the values of the partition markers and the assignment of local elements
to processes.
To determine the MPI sender and receiver pairs, we compare the element counts
$\arrE$ before and after but may ignore all other data fields inside the forest
objects.
The messages sizes follow from $\arrE$ as well.
%The markers $\arrm$ may be recreated after transfering the elements using a
%call to \mpiallgather.
Thus, the fixed size data transfer is algorithmically similar to the transfer
of elements during partitioning.
We refer to this operation as
\begin{equation*}
  %\eqnlab{transferfixed}
  \text{\pftransferfixed
    ($\arrE$ before/after, data array before/after, data size).}
\end{equation*}
Note that it is possible to split it into a begin/end pair to perform
computation while the messages are in transit.
In practice, we proceed along the lines of \algref{fixedrepartitionscheme}.
% to
% partition a forest and transfer the data to the new partition at the same time
% \cite{WilcoxStadlerBursteddeEtAl10}.%
%
\begin{algorithm}
  \caption{fixed size data transfer (forest $f$, data $d_\rbefore$)
           $\to$ data $d_\rafter$}
  \alglab{fixedrepartitionscheme}
  One recommeded procedure to repartition per-element data in application memory.
  \newline
  \begin{algorithmic}[1]
    \STATE $\arrE_\rbefore \leftarrow f.\arrE$
      \hfill\COMMENT{deep copy element counts before partition}
      \linelab{fixedone}
    \STATE \pfpartition ($f$)
      \hfill\COMMENT{modify members of forest in place}
    \STATE $\arrE_\rafter \leftarrow f.\arrE$
      \hfill\COMMENT{reference counts after partition}
      \linelab{fixedtwo}
    \STATE $d_\rafter \leftarrow$ allocate fixed size data ($f$)
      \hfill\COMMENT{layout known from forest}
%     \STATE context $c \leftarrow$ \pffixedbegin
%       ($\arrE_\rbefore$,
%        $\arrE_\rafter$,
%        $d_\rbefore$,
%        $d_\rafter$,
%        size (data)%
%        )
%     \STATE
%       \hfill\COMMENT{do some work while messages are in transit}\hfill\mbox{}
%     \STATE \pftransferend ($c$)
%       \hfill\COMMENT{wait for message completion; free $c$}
     \STATE \pftransferfixed
       ($\arrE_\rbefore$,
        $\arrE_\rafter$,
        $d_\rbefore$,
        $d_\rafter$,
        size (element data)%
        )
    \STATE free ($d_\rbefore$) ; free ($\arrE_\rbefore$)
      \hfill\COMMENT{memory no longer needed}
  \end{algorithmic}
\end{algorithm}

% This is the version specialized to fixed size data.
        % Note that we split the transfer into a begin and end call to allow for overlap
% of communication and computation in the meantime.

% p4est_transfer_context_t *p4est_transfer_fixed_begin (const p4est_gloidx_t *
%                                                       dest_gfq,
%                                                       const p4est_gloidx_t *
%                                                       src_gfq,
%                                                       sc_MPI_Comm mpicomm,
%                                                       int tag,
%                                                       void *dest_data,
%                                                       const void *src_data,
%                                                       size_t data_size);

% \todo{Describe fixed size transfer in some detail}

When the data size varies between elements, we propose to store the sizes in an
array with one integer entry for each local element.
As with the fixed size, the data itself is contiguous in memory in ascending
order of the local elements.
A non-redundant implementation calls the fixed size transfer with the
array of sizes to make the data layout available to the destination processes.
With this information known, the memory for the data after partition is
allocated in another contiguous block and the transfer for the data of variable
size executes.
We have implemented this generalized communication routine as
\begin{equation*}
  %\eqnlab{transfervariable}
  \text{\pftransfervariable
    ($\arrE$ before/after, data before/after, sizes before/after).}
\end{equation*}
Thus, we pay a second round of asynchronous point-to-point communication for
the benefit of code simplicity and reuse.
Alternatively, it would be possible to rewrite the algorithm using a polling
mechanism to minimize wait times at the expense of CPU load.
The listing for the combined partition and transfer is
\algref{variablerepartitionscheme}.
\begin{algorithm}
  \caption{variable size data transfer
    \newline\mbox{}\hfill
    (forest $f$, data $d_\rbefore$, sizes $s_\rbefore$)
    $\to$ (data $d_\rafter$, sizes $s_\rafter$)}
  \alglab{variablerepartitionscheme}
  When application data size varies by element, we propose to repartition
  just the size information first, as if this were user data, which is then
  sufficient to call the variable-payload transfer function.
  Both rounds use point-to-point messages with known receivers, ranks, and
  sizes, which optimizes buffer space and eliminates setup.
  \newline
  \begin{algorithmic}[1]
    \STATE
      \hfill\COMMENT{%
      partition as in \algref{fixedrepartitionscheme},
      Lines~\ref{line:fixedone}--\ref{line:fixedtwo}}%
      \hfill\mbox{}
    % \STATE $\arrE_\rbefore \leftarrow f.\arrE$
    %   \hfill\COMMENT{deep copy counts before partition}
    % \STATE \pfpartition ($f$)
    %   \hfill\COMMENT{modify members of forest in place}
    % \STATE $\arrE_\rafter \leftarrow f.\arrE$
    %   \hfill\COMMENT{reference counts after partition}
    \STATE $s_\rafter \leftarrow$ allocate array of sizes ($f$)
      \hfill\COMMENT{layout known from forest}
%    \STATE context $c_s \leftarrow$ \pffixedbegin
%      ($\arrE_\rbefore$,
%       $\arrE_\rafter$,
%       $s_\rbefore$,
%       $s_\rafter$,
%       size (integer)%
%       )
%    \STATE
%      \hfill\COMMENT{do some work while messages are in transit}\hfill\mbox{}
%    \STATE \pftransferend ($c_s$)
%      %\hfill\COMMENT{wait for message completion; free $c_s$}
    \STATE \pftransferfixed
      ($\arrE_\rbefore$,
       $\arrE_\rafter$,
       $s_\rbefore$,
       $s_\rafter$,
       size (integer)%
      )
    \STATE $d_\rafter \leftarrow$ allocate variable size data ($f$, $s_\rafter$)
      %\hfill\COMMENT{layout known from forest}
%    \STATE context $c_d \leftarrow$ \pfvariablebegin
%      ($\arrE_\rbefore$,
%       $\arrE_\rafter$,
%       $d_\rbefore$,
%       $d_\rafter$,
%       $s_\rbefore$,
%       $s_\rafter$)
%    \STATE
%      \hfill\COMMENT{do some work while messages are in transit}\hfill\mbox{}
%    \STATE \pftransferend ($c_d$)
%      %\hfill\COMMENT{wait for message completion; free $c_d$}
    \STATE \pftransfervariable
      ($\arrE_\rbefore$,
       $\arrE_\rafter$,
       $d_\rbefore$,
       $d_\rafter$,
       $s_\rbefore$,
       $s_\rafter$)
    \STATE  free ($s_\rbefore$) ; free ($d_\rbefore$) ; free ($\arrE_\rbefore$)
      \hfill\COMMENT{memory no longer needed}
  \end{algorithmic}
\end{algorithm}

%\subsubsection{Algorithmic description}

% \subsection{Reversing the communication pattern}
% \seclab{notify}

\subsection{Reversing the communication pattern}
\seclab{notify}

Standard element-based numerical methods lead to a symmetric communication
pattern, that is, every sender also receives a message and vice versa.
The data sent per element is most often of fixed size, thus every process is
able to specify the message size in a call to say \mpiirecv.
In other applications,
% the data distribution between the elements is irregular.
the communication pattern may no longer be symmetric, which means that the
receiver processes have to be notified about the senders.
% In addition, if the data size per element is variable, we also have to inform
% the receiver about the message sizes in order to transfer them when
% repartitioning the mesh.

Pattern reversal can be understood as the transposition of the sender-receiver
matrix, which is an operation available from parallel linear algebra packages;
see e.g.\ \cite{MirzadehGuittetBursteddeEtAl16}.
When trying to minimize code dependencies, we may ask about an efficient way
to code the reversal ourselves.
A parallel algorithm based on a binary tree has been discussed in
\cite{IsaacBursteddeGhattas12}.
Without going into detail, we propose an extension that uses an $n$-ary tree,
where the number of children at each level is configurable, to reduce the
depth and thus the latency of the operation.
The branching of the tree can be configured to match any NUMA/multicore
achitecture.
Futhermore, we have extended this algorithm to carry a payload without
intreasing the number of messages, which is useful to communicate the message
sizes to the receivers.
%TODO: cite Hari's SC paper (2012?)
We will refer to this algorithm as \narynotify.

\section{Demonstration: parallel particle tracking}
\seclab{particles}

%\subsection{Element-based particle tracking}

To exercise the algorithms introduced above, we present a particle tracking
application.
The particles move independently of each other by a gravitational attraction
to several fixed-position suns, following Newton's laws.
Each particle is assigned to exactly one quadrant that contains it and, by
consequence, to exactly one process.
The mesh dynamically adapts to the particle positions by enforcing the rule
that each element may contain at most $E$ particles.
If more than this amount accumulate in any given element, it is refined.
If the combined particle count in a family of leaves drops below $E / 2$, they
are coarsened into their parent.
The features used by this example are:
\begin{itemize}
  \item Explicit Runge-Kutta (RK) time integration of selectable order:
        We use schemes where only the first subdiagonal of RK coefficients
        is nonzero, thus we store just one preceding stage.
        This applies to explicit Euler, Heun's methods of order 2 and 3 and the
        classical RK method of order 4.
  \item Weighted partitioning \cite{BursteddeWilcoxGhattas11}:
        Each quadrant is assigned the weight approximately proportional to
        % $1 + e$, where $e$ is
        the number of particles it contains.
        This way the RK time integration is load balanced between the processes.
        %and we do not
        %exhaust the memory if we have many empty quadrants.
  \item Partition traversal (\secref{traverse}):
        In each RK stage, the next evaluated positions of the local particles
        are bulk-searched in the partition.
        If found on the local process, we continue
        a local search to find its next local owner quadrant.
        If found on a remote process, we send it to that process for the
        next RK stage.
  \item Reversal of the communication pattern (\secref{notify}):
        A process does not know from which processes it receives new particles,
        thus we call the $n$-ary notify function to determine the \mpiirecv
        operations we need to post.
  \item Variable-size parallel data transfer on partitioning
        (\secref{transfer}):
        Since the amount of particles per quadrant varies, we send
        variable amounts of per-element data from the old owners to the new.
  \item Construction of a sparse forest (\secref{build}):
        At selected times of the simulation, we use a small subset of particles
        to build a new forest, where each of the selected particles is placed
        in a quadrant of a given maximal level.
        The rest of this forest is filled with the coarsest possible quadrants.
        Depending on the setup, it has less elements and is
        thus better suited for offline post-processing or visualization.
  \item Partition-independent I/O (\secref{pertree}):
        We compute the cumulative per-tree element counts for both the
        current and the sparse forests.
\end{itemize}

\subsection{Simulation setup}
\seclab{simsetup}

The problem is formulated in the 3D unit cube $[0, 1]^3$.
We mesh it with one tree except where explicitly stated.
If a particle leaves the domain, it is erased, thus the global number may drop
with time.
The three suns are not moving.
The initial particle distribution is Gau\ss{}-normal.
Each particle has unit mass and initial velocity $0$ and the gravitational
constant is $\gamma = 1$; see \tabref{suns} for details.
\begin{table}
\begin{center}
\begin{tabular}{ccc|c}
  $x$ & $y$ & $z$ & mass \\
  \hline
  .48 & .58 & .59 & .049 \\
  .58 & .41 & .46 & .167 \\
  .51 & .52 & .42 & .060 \\
\end{tabular}
\hspace{5ex}
\begin{tabular}{c|c}
  \multicolumn{2}{c}{particle distribution (Gau\ss)} \\
  \hline
  center & $\mu = (.3, .4, .5)$ \\
  standard deviation & $\sigma = .07$ \\
\end{tabular}
\end{center}
\caption{The three suns (left)
         and the parameters of the initial particle distribution (right).}%
\label{tab:suns}%
\end{table}%

% Each particle interacts with one or more fixed-position point masses
% and moves according to Newton's laws.

The parameters of a simulation include the global number of particles, the
maximum number $E$ of particles per element, minimum and maximum levels of
refinement, the order of the RK method, the time step $\Delta t$ and the final
simulated time $T$.
% For weak scaling experiments, we increase the process count $P$ by powers of
% $8$, simultaneously increasing the level parameters by one and dividing $\Delta
% t$ by 2.

% Rather initialize based on kinetic and gravitational energy?  No.

The initial particle distribution and mesh are created in a setup loop.
Beginning with a minimum-level uniform mesh, we compute the integral of the
initial particle density per element and normalize by the integral over the domain.
We do this numerically using a tensor-product two-point Gau\ss{} rule.
From this, we compute the current number of particles in each element, compare
it with $E$ and refine if necessary.
After refinement, we partition and repeat the cycle until the loop terminates
by sufficient refinement or the specified maximum level is reached.
Only then we allocate the local particles' memory and create the particles
using per-element uniform random sampling.
Thus, neither the global particle number nor their distribution is met exactly,
but both approach the ideal with increasing refinement.

%Move particles by ODE solver.
%Runge-Kutta solver: Need evaluations at multiple locations.
%  Newton: does not require search
%  Analytic velocity field: does not require search
%  FE velocity field: requires search to interpolate velocity
%Search required to deal with large time steps: where does the particle go?
%
%With variable speeds this is interesting, since step length varies wildly
%
%Coarsen and refine to equalize number of particles per element
%Partition, transfer particles
%variable data size: proportional to number of particles in element

To make the test on the AMR algorithms as strict as possible, the parallel
particle redistribution and the mesh refinement and partitioning occur once in
each stage of each RK step.
We choose the time step $\Delta t$ proportional to the characteristic element
length to establish a typical CFL number.
Thus, we may create a scaling series of increasing problem size (that is,
particle count and resolution) at fixed CFL.
Our non-local particle transfer is designed to support arbitrarily large CFL,
where the amount of senders and receivers for each process effectively depends
on the CFL only, even if the problem size is varied by orders of magnitude.

We run each series to a fixed final time $T$, which produces a certain
distribution of the particles in space (see \figref{plot7r4b}).
\begin{figure}
\begin{center}
  \includegraphics[width=.9\columnwidth]{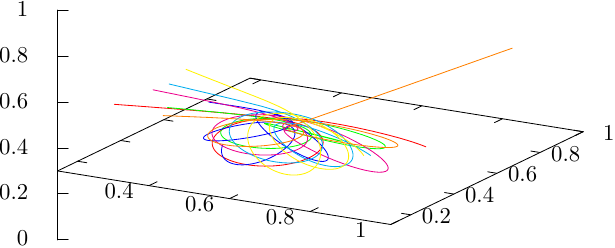}
\end{center}
\caption{%-p -T2 -e5 -Y4:4:16 -r3 -e2 -n50 -h.003 -l3 -L8 -R7
  Trajectory of seven out of 44 particles tracked to time $T = 2$ with
  the fourth-order RK method and $\Delta t = .003$.
  The initial positions of the particles are visible on the left hand side.}%
\label{fig:plot7r4b}%
\end{figure}%
The number of time steps required doubles with each refinement level.
To allow for a meaningful comparison between different problem sizes,
we measure the wall clock times for the RK method and all parallel algorithms
in the final time step, averaging over the RK stages.
We compute the per-tree element counts and the sparse forest at selected times
of the simulation (see \figref{plot7cut}), where we only use the timing of the
last one at $T$.
\begin{figure}
\begin{center}
  \includegraphics[width=.43\columnwidth]{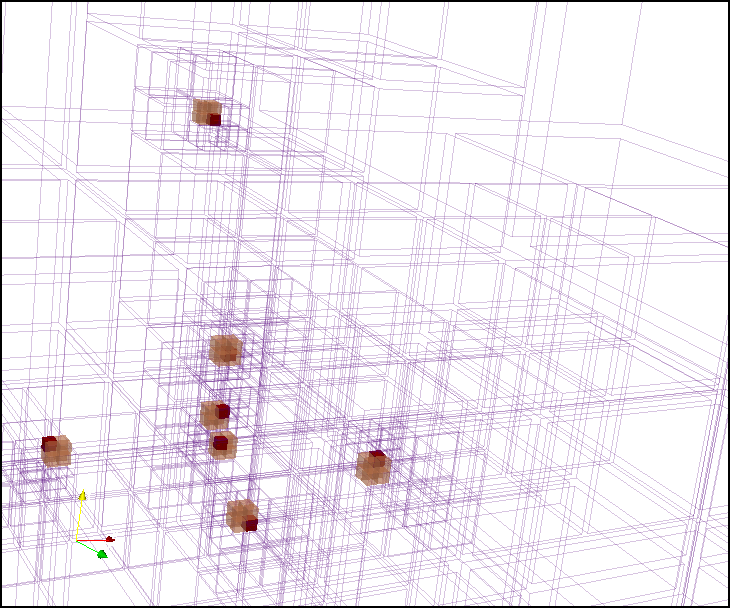}
  \hspace{2ex}
  \includegraphics[width=.43\columnwidth]{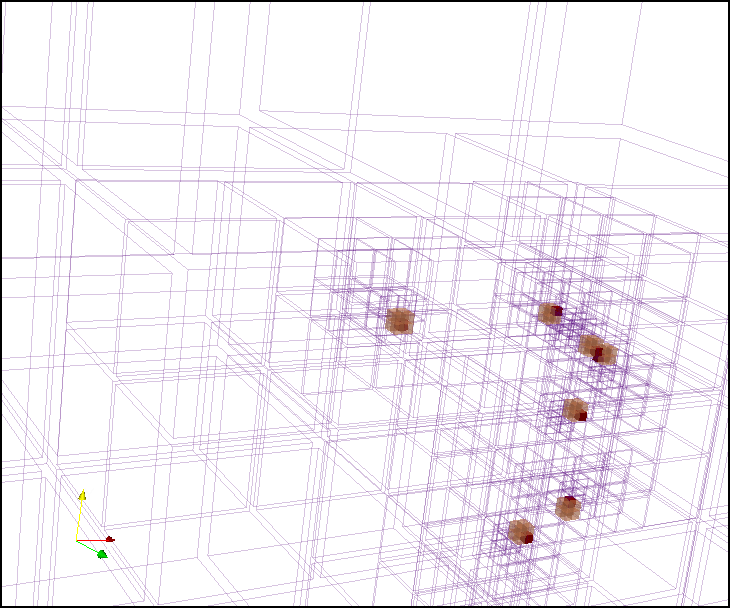}
\end{center}
\caption{Zoom into sparse forests created at time $t = 0$ (left) and $t = .5$
  (right), respectively, using the the same setup as for \figref{plot7r4b},
  here with $\Delta t = .002$.
  Of the 44 particles tracked, the same seven are added to both sparse forests
  as individual level-8 elements (cf.\ \algref{buildadd}).
  Elements up to level 6 are drawn as blue wireframe, level 7 as
  transparent orange and level 8 as solid red.}%
\label{fig:plot7cut}%
\end{figure}%

%\todo{Comment on realism}

% 2D version: keep all 3D storage and calculations, but project coordinates onto
% the 2D manifold.

We use three problem setups of increasing overall particle count and CFL, which
we run to $T = .4$ with the 3rd order RK scheme (see \tabref{setup}).
We use process counts from 16 to 65,536 in multiples of eight, which matches the
multiplier of the particle counts and in consequence that of the element
counts.
The computations reported in this section have been run on Juqueen, an IBM
BlueGene/Q system with 28 racks and a total of 28,672 16-core, 16 GB compute
nodes (458,752 CPU cores) of type IBM PowerPC A2 (1.6 GHz) connected by a
5D-torus network with a dedicated collective subnet \cite{Juqueen}.
\begin{sidewaystable}
\begin{center}
\begin{tabular}{c|ccc|ccc|ccc}
         & \multicolumn{3}{c|}{particles}
         & \multicolumn{3}{c|}{elements}
         & \\
  levels & \#req & \#eff & \#end &
           level & \#eff & \#end &
           $\Delta t$ & \#steps & \#peers
  \\
  \hline
  3--9  &    12800 &    13318 &    12917 &
                 7 &     8548 &    12762 &  .008 &  50 & 5.31 \\
  5--11 &   819200&    852580 &   842250 &
                 9 &   538392 &   784848 &  .002 & 200 & 6.96 \\
  7--13 & 52428800 & 54513360 & 54283090 &
                11 & 34418420 & 49821220 & .0005 & 800 & 7.12 \\
  \hline
  3--9  &    102400 &    102374 &     98359 &
                  6 &      1632 &      2059 & .016 &  25 & 8.06 \\
  5--11 &   6553600 &   6553472 &   6424887 &
                  8 &    102068 &    131587 & .004 & 100 & 11.5 \\
  7--13 & 419430400 & 419415934 & 416393854 &
                 10 &   6530210 &   8700623 & .001 & 400 & 11.1 \\
  \hline
  4--10 &     5120000 &     5119830 &     4935040 &
                    8 &       55434 &       79486 & .016 & 25 & 10.7 \\
  6--12 &   327680000 &   327677582 &   321323140 &
                   10 &     3528876 &     6366600 & .004 & 100 & 22.6 \\
  8--14 & 20971520000 & 20971146506 & 20820237439 &
                   12 &   225760858 &   353507330 & .001 & 400 & 19.7 \\
\end{tabular}
\end{center}
\caption{Three problem sizes, each run with process counts from 16 to 65,536
         in multiples of 8.  We only show every other run (multiple of 64).
         The problems have maximum particle counts per element $E$ of 5 and
         twice 320, respectively.
         For each run, we provide the specified minimum and maximum levels, the
         particle counts referring to the initial request, the count
         effectively reached on initialization, and the count at $T = .4$,
         respectively.
         Over time, we lose some particles that leave the domain.
         For the elements, we show the initial maximum level and
         global count and the count at final time $T$.
         Over time, we create more elements since they move
         closer together at $t = .4$, which leads to a deeper tree.
         On the right, we show the time step size, number of steps,
         and the average number of communication peers for particle transfer.
         The CFL number increases between the three problem sizes,
         which can bee seen by comparing the levels with $\Delta t$,
         and reflects in \#peers.
         The overall largest run creates 20.97 billion particles.}%
\label{tab:setup}%
\end{sidewaystable}%

\subsection{Load balance}
\seclab{simload}

We know that \pforest has a fast partitioning routine to equidistribute the
elements between the processes \cite{BursteddeHolke16b}.
Here we need to equidistribute the load of the RK time integration, which is
proportional to the local number of particles.
To this end, we assign each element a weight $w$ for partitioning that derives
from the number of particles $e$ in this element, $w = 1 + e$.
We offset the weight by 1 to bound the memory used by elements that contain
zero or very few particles.
We test the load balance by measuring the RK integration times in a weak
and strong scaling experiment.
From \figref{plotrk} we see that scalability is indeed close to perfect.
%
% Scalablility of RK step with second example
%
\begin{figure}
\begin{center}
\includegraphics[width=.49\columnwidth]{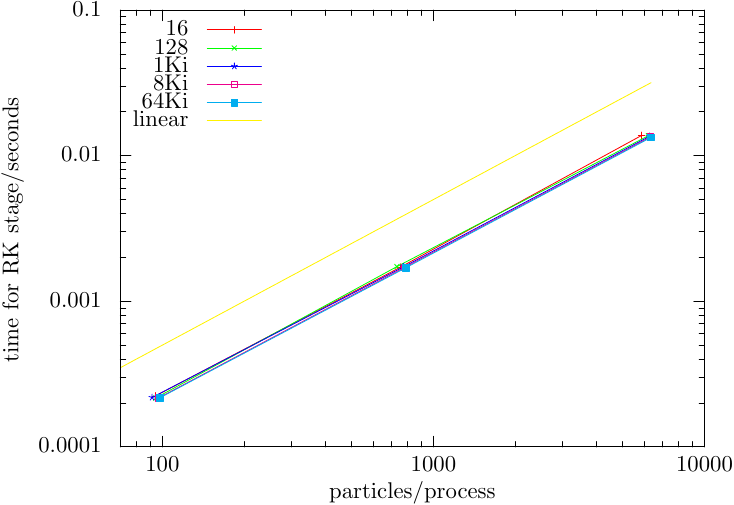}
\hfill
\includegraphics[width=.49\columnwidth]{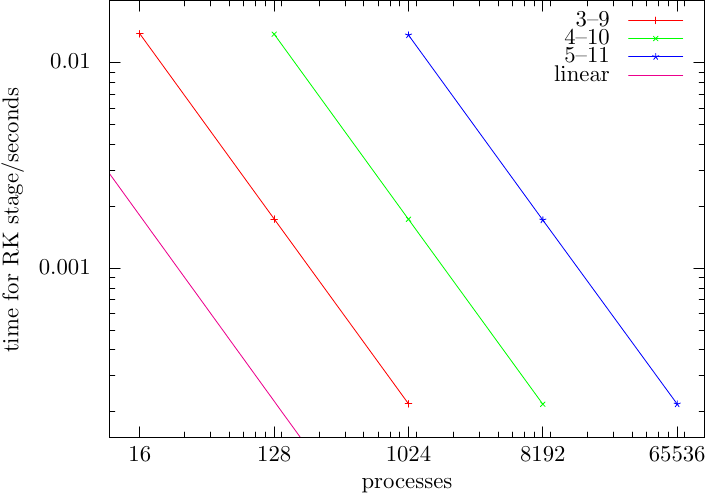}
\end{center}
\caption{Scaling of the Runge-Kutta time integration.
  % shows the average time of
  %one RK stage between three total in the last time step.
  We use the mid-size problem from \tabref{setup} and rerun each line
  with $8\times$ and $64\times$ processes (equivalently, rerun the
  $8\times$ and $64\times$ smaller problems with the same process count),
  hence three dots per line.
  Left: The number of MPI processes is color-coded.
  We confirm optimal weak scaling since the dots lie on top of each other
  and optimal strong scaling by the fact that the lines lie on top of each
  other and have unit slope.
  Right: A typical strong scaling diagram, indicating simulation size by the
  levels of refinement.
  These plots indicate successful load balance by the particle-weighted
  partitioning of elements.
}%
\label{fig:plotrk}%
\end{figure}%

%Examples: Discuss CFL

A weight function that counts both elements and particles in some ratio has
been proposed before
% \cite{GassmoellerHeienPuckettEtAl16},
\cite{GassmoellerLokavarapeHeienEtAl18},
as has the
initialization of particles based on integrating a distribution function.
In the above reference, parallelization is based on a one-element ghost layer.
The use of algorithms like ours for non-local particle transfer and variable
data, as we describe it below, has not yet been covered as far as we know.

\subsection{Particle search and communication}
\seclab{simcomm}

We use the top-down forest tra\-ver\-sal \algref{searchpartition},
\pfsearchpartition, augmented with a local search to determine for each local
particle whether it changes the local element or leaves the process domain.
In the first case, we find this element, and in the latter case, we find which
process it is sent to.
Once we know this, we reverse the communication pattern using \narynotify to
inform the receivers about the senders and send the particles using
non-blocking MPI.
We are not using one-sided MPI, since synchronization is often slower than
messaging itself \cite{GerstenbergerBestaHoefler18}, which would defeat the
purpose in our case.

Moving particles between elements is followed by mesh coarsening and
refinement, which generally upsets the load balance, so we repartition the
forest.
This changes an individual element's ownership, and thus the contained
particles' ownership, from one process to another.
Thus, we transfer the particles a second time, now from the old to the new
partition.
%
% Transfering the particles
% % to their new owners
% after partition
We use the two-stage \algref{variablerepartitionscheme}, where we first
send the number of particles for each element (fixed-size message volume per
element) and then send the particles themselves (variable-size volume).

According to our measurements, \narynotify has runtimes well below or around
1~ms for the small- and mid-size problems.
% except for the two biggest runs at 64Ki processes,
%which require 1.0~ms and 1.2~ms, respectively.
The large problem gives rise to runtimes of about 5~ms.
The fixed-size particle transfer is a sub-millisecond call.
\begin{figure}
\begin{center}
\includegraphics[width=.49\columnwidth]{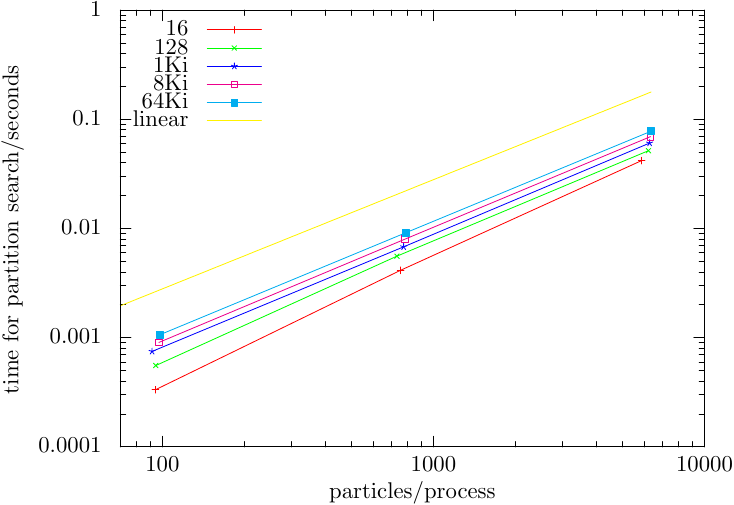}
\hfill
\includegraphics[width=.49\columnwidth]{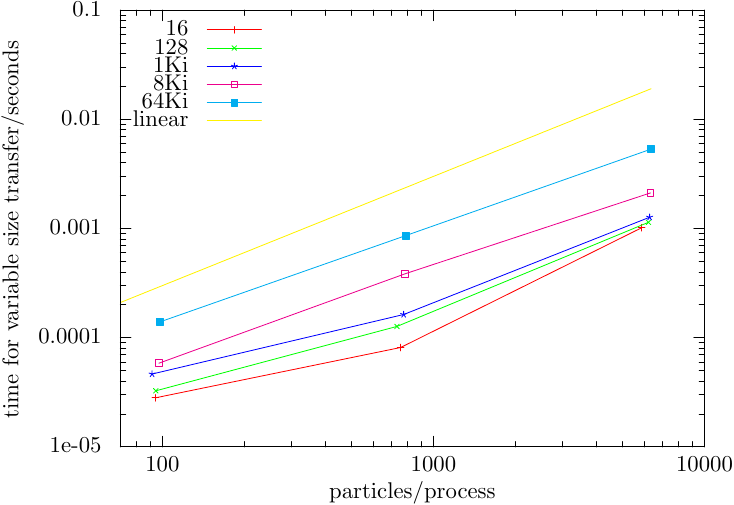}
\end{center}
\caption{Combined partition and local search (left) and transfer of
         variable-size element data (right) for the mid-size problem, where the
         runtimes are measured in the final time step.}%
\label{fig:searchtrans}%
\end{figure}%
Runtimes of the remaining calls \pfsearchpartition and \pftransfervariable
for the mid-size problem are displayed in \figref{searchtrans}.
Their scalability is generally acceptable given their small absolute runtimes.
\begin{table}
\begin{center}
\begin{tabular}{c|c|c|c}
  $P$ & small & medium & large \\
  \hline
     16 & 9.29e-3 & 41.9e-3 & 3.12 \\
    128 & 10.5e-3 & 51.6e-3 & 3.63 \\
   1024 & 11.6e-3 & 60.6e-3 & 4.13 \\
   8192 & 12.8e-3 & 69.4e-3 & 4.62 \\
  65536 & 13.9e-3 & 77.9e-3 & 5.10
\end{tabular}
\\[2ex]
\begin{tabular}{c|cccc}
  $P$ / $K$ & 1 & 8 & 64 & 512 \\
  \hline
16    & 9.29e-3 & 9.05e-3 & 13.9e-3 & 58.4e-3 \\ %
1024  & 11.6e-3 & 11.4e-3 & 16.3e-3 & 61.8e-3 \\ % 0.429102
65536 & 13.9e-3 & 13.7e-3 & 18.8e-3 & 66.2e-3 \\ % 0.440107
\end{tabular}
\end{center}
\caption{Top:
         Absolute runtimes in seconds of \pfsearchpartition augmented with a
         local search for the three problem sizes from \tabref{setup}.
         Each column presents a weak scaling exercise, where ideal times would
         be constant.
         The three runs have comparable rates between 60k and 82k particles per
         second.
         % (value for $P = 64\mathrm{Ki}$).
         Bottom:
         We use a forest with $K$ trees in a cubic brick layout,
         where the refinement in each tree is reduced accordingly to make the
         meshes identical (shown for the small problem).
         For roughly a hundred trees and above the run times increase with $K$
         while remaining largely independent of the process count $P$.
}%
\label{tab:prob3searchp}%
\end{table}%
The runtimes of \pfsearchpartition for all three problem sizes are compared in
\tabref{prob3searchp}.
They grow by less than a factor of 2 in weak scaling while increasing the
process and particle counts by more than three orders of magnitude.
In this test, we also experiment with forest meshes of up to $K = 2^{d \times
B}$ trees, where $B$ runs from 0 to 3 and per-tree minimum and maximum levels
decrease by $B$, which keeps the meshes identical independent of $K$.
Since the forest connectivity is unstructured, the limit of many trees loses
the hierarchic property of the mesh, which reflects in a slower search.
Up to 512 trees we see search times of less than 1/10th seconds for the small
problem.
For $1/8$th of the large problem (not shown in the table), the search times
increase by a factor between 7 and 10 from 1 to 512 trees (.32 seconds on $K =
1$, $P = 16$ to 3.43 seconds on $K = 512$, $P = 64\mathrm{Ki}$).

\subsection{Sparse forest and per-tree counts}
\seclab{simbuild}

At the end of the simulation, we create a sparse forest for output and
post-processing.
We use every 100th particle for the small size problem and every 1000th particle
for the medium and large size problems; let us call this factor $R \ge 1$.
The ratio of $E$ and $R$ and the specified maximum level determine the size of
the sparse forest.
If the maximum level is high, we create a deeper forest and more elements
compared to the simulation.
If $E/R$ is one, we keep the number of elements roughly the same, if it is
less than one, the sparse forest will have less elements.
These two effects may offset each other.
In our examples, the sparse forest is smaller in the small-scale problem and
larger in the mid- and large size problems.
The build times of the largest run for each problem setup are
4.8~ms for the small, 20.5~ms for the medium, and 358~ms for the large size
problem, each obtained with 65,536 MPI processes.
Especially for the two larger problems, we have much less elements than
particles, such that the number of elements per process is in the aggressive
strong scaling regime.

% Build: do proper weak and strong scaling for mid-size run:  No.

The global per-tree counting of elements has runtimes below or around 1~ms
except for the runs on 65,536 processes, where it is 4.4~ms for all three
problem setups (using one tree).
When reproducing the same mesh with a brick forest of as much as 512 trees,
the run times do not change in any significant way.
Since the messages are sent concurrently (the algorithm avoids daisy-chaining),
this is achieved by design.
This function has been tested in even more varied situations by the community
for several years (transparently through \pfsave).

%This is explained by the fact that the local computation is negligible by the
%design of \pfpertree, and there is one constant-size message sent per process
%independently of the number of elements or particles.

%\subsection{Parallel generation of a porous medium}
%\seclab{porous}
%
%\subsection{An in-situ visualization pipeline}
%\seclab{parvis}

\section{Demonstration: constructing random spheres}
\seclab{spheres}

Particles, as considered in the previous demonstration, have extent zero and
are only stored on one process at any given time.
Now let us consider objects of non-trivial extent, for example geometric
objects such as spheres.
Depending on refinement and partition, and the individual spheres' properties,
some or even most spheres cover a region in space that is split between several
processes.
We require the global search proposed in \secref{traverse}, the variable-size
data transfer from \secref{transfer} and the pattern reversal from
\secref{notify} to construct this model in parallel.

Our objective is to create multiple spheres based on a pseudorandom
distribution and to refine the mesh at the spheres' boundaries.
The mesh refinement shall be reproducible and partition-independent.
Such a setup may model a porous medium, where the spheres represent obstacles
whose surface must be accurately resolved for a flow simulation in the empty
space.
Conversely, the spheres may be hollow and we increase their density to simulate
percolation, treating the empty space as the solid matrix.
Lastly, problems exist when only the surface of the spheres is of interest,
for example in visualization applications \cite{Burstedde18b}.

\subsection{Construction procedure}
\seclab{sphere-construct}

The construction has three parts.
First, each process creates a certain number of spheres with centers inside its
partition.
For these spheres, the process becomes the current owner.
Second, all remote processes intersecting an owned sphere's surface are
determined by the partition search, the spheres' metadata is transferred to
each such process, which we then use to decide where to refine the mesh.
After refinement, we discard the copies.
Third, we repartition the mesh, and the current owner of a sphere sends its
metadata to its new owner.
Steps two and three are repeated in a loop over increasing refinement levels to
ensure that both the computational load and the memory consumption are well
balanced.
% Naturally, we aim to do the communication in steps two and three in bulk and
% include only the required data.

We define a probability distribution of the spheres $\rho$ depending on their
radius.
% Even though the procedure below is not exact for overlapping spheres, it serves
% well for our illustration.
To make all radii equally likely in a given volume, let
\begin{equation}
  \eqnlab{spheredist}
  \rho (r) = c / r^3 , \quad r_\rmin \le r \le r_\rmax,
  \quad \text{and $0$ otherwise}.
\end{equation}
Normalization makes $c$ an expression in $r_\rmin$ and $r_\rmax$, and the
expected volume is
\begin{equation}
  V_E = \int_{r_\rmin}^{r_\rmax} \frac43 \pi r^3 \rho (r) \dr
      = \frac43 \pi \frac{r_\rmin^2 r_\rmax^2}{(r_\rmin + r_\rmax) / 2}
      = \frac43 \pi \frac{r_\rgeom^4}{r_\rarith}
  .
\end{equation}
To enforce an overall volume density $q$, an element of volume $V_e$ must have
\begin{equation}
  N_e = q V_e / V_E
\end{equation}
spheres on average, which we realize by sampling the number of constructed
spheres separately for each element from a Poisson distribution with mean
$N_e$.
We sample the center coordinates of each sphere uniformly in the element and
draw its radius from \eqnref{spheredist}.
This process is independent between elements, the only issue being the
initialization of the random number generator in parallel.
We resolve it by seeding the generator anew for each element with this
elements' lower left octree coordinates, which makes the distribution
reproducible and partition-independent.
\begin{figure}%
  \begin{center}
  \includegraphics[height=52mm]{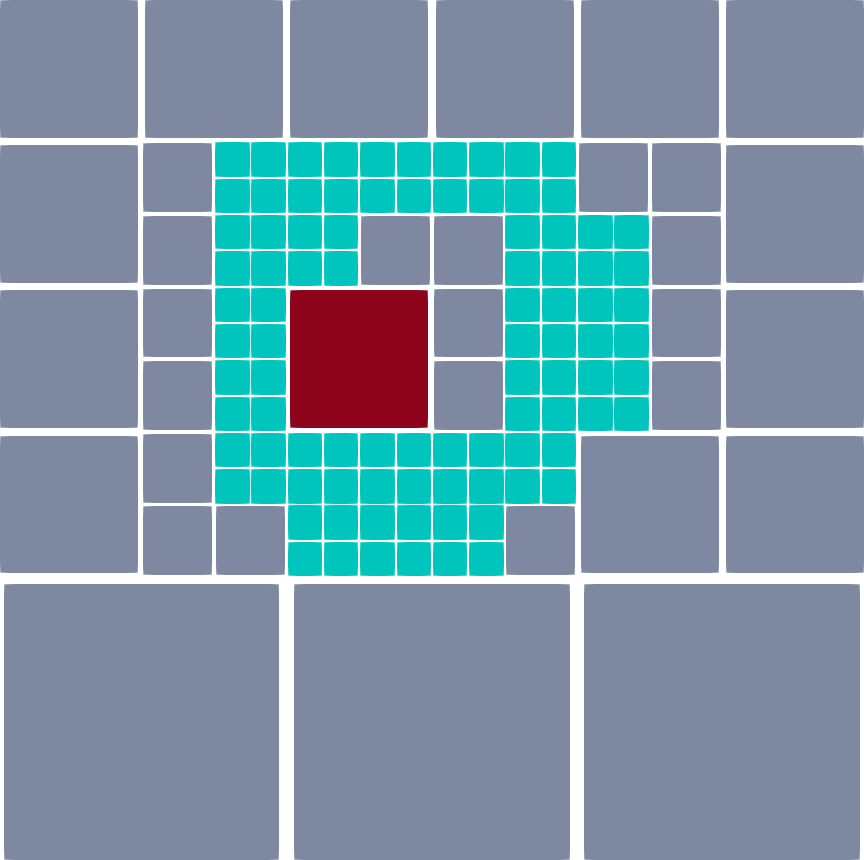}%
  \hfil
  \includegraphics[height=52mm]{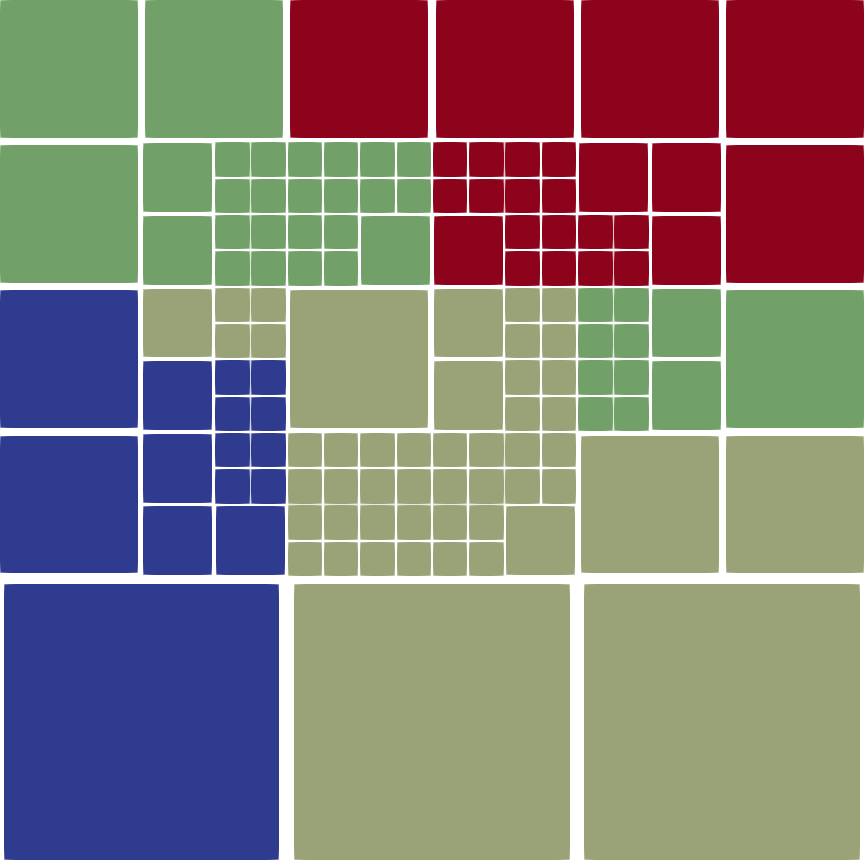}%
  \end{center}
  \caption{Pseudo-random, reproducible, and partition-independent resolution of
           spheres (zoom into 2D example).
           On the left, we see the element intersecting the sphere's center
           (red) and the refinement to its surface (green).
           On the right, we color by four MPI processes used and note that
           the refinement is non-local relative to the sphere's owner quadrant
           and process.
           The refinement remains the same when using different numbers of
           processes (not shown).
           We do not enforce 2:1 balance; an element is refined if its
           parent intersected a sphere's surface at some time during the
           refinement loop.
           See \cite{Burstedde18b} for a continued discussion of this example
           in the context of in-situ visualization.%
  }%
  \figlab{sphere-construct}%
\end{figure}%}
\figref{sphere-construct} shows a typical construction along with the mesh.

\subsection{Computational experiments}
\seclab{sphere-timing}

A scaling series can be devised by relating $r_\rmin$ and $r_\rmax$ to the
minimal and maximal refinement levels and the MPI process count.
If we divide both radii by two, the expected volume shrinks and the total
number of spheres grows, each by a factor of eight in 3D.
If, at the same time, we increase both minimal and maximal levels $\ell_\rmin$,
$\ell_\rmax$ by 1, we achieve a constant ratio of spheres to elements.
A weak scaling series emerges if we scale the number of processes by multiples
of eight.

There is one parameter left to choose, namely the desired number of elements
$s$ spanning a sphere's radius.
This parameter stops the refinement for larger spheres earlier than for
smaller ones and ensures that the ratio of spheres to elements stays truly
constant.
% in the scaling described above.
We show some typical numbers for $s = 4$ in \tabref{spheres-s4} on up to
24,576 processes of the new Juwels supercomputer at JSC.
We have used the ``standard'' compute nodes of Juwels, comprising dual Intel
Xeon Platinum 8168 processors with 24 2.7 GHz cores each at 2 GB RAM per core,
connected by EDR-Infiniband \cite{JuwelsWeb}.
\begin{table}
  \begin{center}
    \begin{tabular}{r|rrrr|rr|rr}
      \multicolumn{1}{c|}{$P$} & $\ell_\rmin$ & $\ell_\rmax$
                               & $r_\rmin$ & $r_\rmax$
                               & \multicolumn{1}{c}{spheres}
                               & \multicolumn{1}{c}{elements}
                               & \multicolumn{1}{c}{$T_\mathrm{pt}$[s]}
                               & \multicolumn{1}{c}{$T_\mathrm{l}$[s]}
                               \\
      \hline
      3,072 & 5 & 12 & 488u & 63m & 4.23M & 1.50G &  8m &  20m \\
            & 6 & 13 & 244u & 31m & 33.8M & 12.0G & 40m & 147m \\
      \hline
     24,576 & 7 & 14 & 122u & 16m & 271M & 96.0G &  47m & 153m \\
            & 8 & 15 & 61u & 7.8m & 2.16G & 768G$^\dag$ & $\ast$ & $\ast$ \\
            & 9 & 16 & 31u & 3.9m & 17.3G & $\ast$ & $\ast$ & $\ast$ \\
    \end{tabular}
  \end{center}
  \caption{Scaling of 3D sphere sampling with $s = 4$ elements per
           radius on the Juwels supercomputer.  We use the size suffixes
           $u = 10^{-6}$, $m = 10^{-3}$, $M = 10^6$, $G = 10^9$.
           We observe that sphere and element counts scale by exact powers of
           eight, which confirms the validity of our pseudorandom
           generation for generating the spheres.
           The largest mesh refines to $768 \times 10^9$ elements, where
           the symbol $\dag$ indicates that subsequent partitioning
           exhausts the machine's memory.
           At $\ell_\rmax = 16$, we successfully create $17 \times 10^9$
           spheres but complete refinement only up to level 15.
           The rightmost two columns show the run times of partition search
           and variable-size transfer $T_\mathrm{pt}$ as introduced in this
           paper, and the run-time of the local search $T_\mathrm{l}$ for
           reference.
           We list the times in milliseconds, observing that the
           non-local algorithms require less than 1/20th of a second
           up to level 14.%
          }%
  \tablab{spheres-s4}%
\end{table}%

% TODO: discuss notify

\section{Conclusion}
\seclab{conclusion}

% This paper targets applications beyond standard simulation.
%
% Applicability for wider range of applications:
% need increased flexibility with respect to multiple aspects.
%
% flexibility/features.
%
This paper provides algorithms that support the efficient parallelization of
computational applications of increased generality.
Such generalization may refer to multiple aspects.
One concerns the location of objects in the partition beyond a one-element ghost
layer, together with flexible criteria for matching and pruning.
Another is the fast repartitioning of variable-sized element data in linear
storage.
When considering the increased importance of scalable end-to-end simulation,
our algorithms may aid in pre-processing (setting up correlated spatial fields
in parallel, or finding physical source and receiver locations) and
post-processing and reproducibility (writing/reading partition-independent
formats of variable-size element data, optionally selecting readapted subsets).

% Variable process counts
% Variable data sizes
%
% Saving/Loading the mesh on different process counts
% One partition-independent file
% Varying data size per element.
% Save/load this.  Partition this.

% We present fairly general solutions in this paper.
%
% Low-level building blocks
% Then combine algorithms of each type as needed.
%
Our algorithms are application-agnostic, that is, they do not
interpret the data or meshes they handle, and perform well-defined
tasks while hiding the complexity of their execution.
Most are fairly low-level in the sense that they reside in the parallelization
and metadata layer of an application.
They can be integrated by third-party libraries and frameworks and often do not
need to be exposed to the domain scientist.
This approach supports modularity, code reuse, and ideally the division of
responsibilities and quicker turnaround times in development.

We draw on the benefits of a distributed tree hierarchy and a linear ordering
of mesh entities.
Without such a hierarchy, the tasks we solve here would be a lot harder or even
impractical (such as the partition search).
We develop all algorithms for a multi-tree forest, noting that they apply
meaningfully to the common special case of a single tree.

% Hitherto: advanced algorithms for dynamic AMR used with fairly
% conventional numerical methods.
% Now add support for a wider range of applications characterized by
% non-uniform data sizes.

% Still, runtimes below one second to create a complete self-contained parallel
% mesh object can be considered fast.

We find that any algorithm runtimes range between milliseconds and a few
seconds, where one second or more occur only for specific algorithms using the
largest setups.
All algorithms are practical and scalable to 21e9 particles and 64Ki MPI
processes on a BlueGene/Q supercomputer system.
In addition, we verify the functionality on the newly installed Juwels
system, creating up to 768e9 elements at a tree depth of 15 levels.

\section*{Acknowledgments}

B.\ gratefully acknowledges travel support by the Bonn Hausdorff Center for
Mathematics (HCM) funded by the Deutsche Forschungsgemeinschaft (DFG, German
Research Foundation) under Germany's Excellence Strategy -- GZ 2047/1,
Project ID 390685813.

The author would like to thank the Gauss Centre for Supercomputing (GCS) for
providing computing time through the John von Neumann Institute for Computing
(NIC) on the GCS share of the supercomputers Juqueen and Juwels at the
J{\"u}lich Supercomputing Centre (JSC).
GCS is the alliance of the three national supercomputing centres HLRS
(Universit\"at Stuttgart), JSC (Forschungszentrum J{\"u}lich), and LRZ
(Bayerische Akademie der Wissenschaften), funded by the German Federal Ministry
of Education and Research (BMBF) and the German State Ministries for Research
of Baden-W{\"u}rttemberg (MWK), Bayern (StMWFK), and Nordrhein-Westfalen (MIWF).

The \pforest software is described on \url{http://www.p4est.org/}.
The source code for the algorithms and the example programs discussed in this
paper is available on \url{http://www.github.com/cburstedde/p4est/}.

We would like to thank A.\ Kraut for sharing her knowledge on Poisson
distributions.
I cannot thank Tobin Isaac enough for inventing \scsplit back in the day.
It is amazing how useful this little algorithm is.

\bibliographystyle{siam}
\bibliography{group,ccgo_new}

\end{document}